\documentclass[sigconf,balance=false]{acmart}
\usepackage{popets}
\usepackage[utf8]{inputenc}
\usepackage{graphicx}
\usepackage{textcomp}
\usepackage{booktabs}
\usepackage{amsmath}
\usepackage{amsfonts}
\usepackage{amsthm}
\usepackage{algorithm}
\usepackage{algpseudocode}
\usepackage{csquotes}
\usepackage{enumitem}
\usepackage{multirow}
\usepackage{multicol}
\usepackage{tabularx}
\usepackage{array}
\usepackage{xcolor}
\usepackage{subcaption}
\usepackage{flushend}
\usepackage{comment}
\usepackage{mathrsfs}
\usepackage{gensymb}
\usepackage{bbold}

\newtheorem{theorem}{Theorem}

\usepackage{hyperref} 
\hypersetup{
    colorlinks=true,
    linkcolor=red,
    filecolor=red,      
    urlcolor=red,
    citecolor=red,
    }

\def \calM{{\mathcal{M}}}
\def \calA{{\mathcal{A}}}
\def \epsdelta{{\epsilon, \delta}}
\def\eg{\emph{e.g.}}
\def\ie{\emph{i.e.}}
\def\cf{\emph{cf.}}
\def\epslb{\epsilon_{emp}}
\def\epsopt{\epsilon_{OPT}}
\def\setOne{\mathbb{1}}

\setcopyright{popets}
\copyrightyear{2024}

\acmYear{2024}
\acmVolume{2024}
\acmNumber{4}
\acmDOI{10.56553/popets-2024-0110}
\acmISBN{}
\acmConference{Proceedings on Privacy Enhancing Technologies}
\begin{document}
\sloppy

\title[Revealing the True Cost of LDP Protocols: An Auditing Perspective]{Revealing the True Cost of Locally Differentially Private Protocols: An Auditing Perspective}
%


\author{Héber H. Arcolezi}
\orcid{0000-0001-8059-7094}
\affiliation{
  \institution{Inria Centre at the University Grenoble Alpes}
  \city{} 
  \state{} 
  \country{France}
}
\email{heber.hwang-arcolezi@inria.fr}

\author{Sébastien Gambs}
\orcid{0000-0002-7326-7377}
\affiliation{
  \institution{Université du Québec à Montréal (UQAM)}
  \city{} 
  \state{} 
  \country{Canada}
}
\email{gambs.sebastien@uqam.ca}


\renewcommand{\shortauthors}{H.H. Arcolezi \& S. Gambs}

\begin{abstract}
While the existing literature on Differential Privacy (DP) auditing predominantly focuses on the centralized model (\eg, in auditing the DP-SGD algorithm), we advocate for extending this approach to audit Local DP (LDP).
To achieve this, we introduce the LDP-Auditor framework for empirically estimating the privacy loss of locally differentially private mechanisms. 
This approach leverages recent advances in designing privacy attacks against LDP frequency estimation protocols. 
More precisely, through the analysis of numerous state-of-the-art LDP protocols, we extensively explore the factors influencing the privacy audit, such as the impact of different encoding and perturbation functions. 
Additionally, we investigate the influence of the domain size and the theoretical privacy loss parameters $\epsilon$ and $\delta$ on local privacy estimation.
In-depth case studies are also conducted to explore specific aspects of LDP auditing, including distinguishability attacks on LDP protocols for longitudinal studies and multidimensional data.
Finally, we present a notable achievement of our LDP-Auditor framework, which is the discovery of a bug in a state-of-the-art LDP Python package.
Overall, our LDP-Auditor framework as well as our study offer valuable insights into the sources of randomness and information loss in LDP protocols. 
These contributions collectively provide a realistic understanding of the local privacy loss, which can help practitioners in selecting the LDP mechanism and privacy parameters that best align with their specific requirements.
We open-sourced LDP-Auditor in~\cite{artifact_ldp_audit}.
\end{abstract}

\maketitle

\section{Introduction} 
\label{sec:introduction}

Differential Privacy (DP)~\cite{Dwork2006} is now widely recognized as the gold standard for providing formal guarantees on the privacy level achieved by an algorithm. 
One of its extension, known as Local DP (LDP)~\cite{first_ldp, Duchi2013}, aims at tackling the trust challenges associated with relying on a centralized server, such as those highlighted by various data breaches~\cite{data_breaches} and instances of data misuse~\cite{cambridge}. 
In LDP, each user perturbs their own data locally before sharing it with a data aggregator or a central server. 
The fundamental idea behind LDP is to introduce carefully calibrated noise to the data to ensure individual privacy guarantees while allowing meaningful statistical analysis to be performed on the aggregated noisy data.

Formally, a randomized algorithm $\calM$ satisfies ($\epsdelta$)-local differential privacy (($\epsdelta$)-LDP), for $\epsilon \geq 0$ and $0 \leq \delta \leq 1$, if for any pair of input values $v_1, v_2 \in \mathrm{Domain}(\calM)$ and all possible sets of outputs $O \subseteq \mathrm{Range}(\calM)$, the following inequality holds:
\begin{equation} \label{eq:ldp}
    \Pr[\calM(v_1) \in O]  \leq e^\epsilon \cdot \Pr[\calM(v_2) \in O] + \delta  \textrm{.}
\end{equation}

In particular, ($\epsdelta$)-LDP is also called approximate LDP, with the special case of $\delta=0$ being called pure $\epsilon$-LDP.
On the one hand, an (L)DP mechanism is accompanied by the mathematical proof in Equation~\eqref{eq:ldp} that establishes a \textbf{theoretical upper bound} for the privacy loss, represented by the privacy parameters $\epsilon$ and $\delta$.
In particular, lower values of $\epsilon$ indicate stronger privacy guarantees.  
On the other hand, the recent and emerging field of DP auditing (\eg, see~\cite{Jagielski2020,nasr2021adversary,Lu2022,tramer2022debugging,maddock2022canife,steinke2023privacy,andrew2023one,pillutla2023unleashing,nasr2023tight,cebere2024tighter,kazmi2024panoramia}) aims at estimating an \textbf{empirical lower bound} for the privacy loss, denoted as $\epslb$.

The role of DP auditing is crucial because it bridges the gap between theoretical guarantees and practical implementations, especially when theoretical bounds on privacy loss might be overly pessimistic or not sufficiently tight (\eg, as in Differentially Private Stochastic Gradient Descent -- DP-SGD~\cite{abadi2016deep}).
In other words, DP auditing helps in understanding how well privacy-preserving mechanisms perform under different conditions and attack scenarios~\cite{nasr2021adversary}.
Furthermore, auditing can uncover potential vulnerabilities or flaws in the implementation that might not be apparent through theoretical analysis alone~\cite{Turati2023,tramer2022debugging,ding2018detecting}.
From a practical standpoint, the empirical estimation of the privacy loss through realistic attackers can also help practitioners make informed decisions and understand the implications of specific privacy parameter choices.
These instances underscore the significance of empirically estimating and verifying the claimed privacy levels of (L)DP mechanisms.

\subsection{Our Contributions}

With these motivations in mind, in this paper, we introduce the LDP-Auditor framework, which is designed to audit LDP frequency estimation protocols and estimate their empirical privacy loss. 
Frequency (or histogram) estimation is a primary objective of LDP as it is a building block for more complex tasks.
\emph{\textit{This means our audit results are applicable and relevant to numerous tasks under LDP guarantees}}, such as heavy hitter estimation~\cite{Bassily2015,Wang2021}, joint distribution estimation~\cite{kikuchi2022castell,Filho2023,Ren2018,Zhang2018}, frequent item-set mining~\cite{Wang2018,Wu2023}, machine learning~\cite{Chamikara2020,Yilmaz2020}, frequency estimation of multidimensional data~\cite{wang2019,nguyen2016collecting,Arcolezi_rs_fd} and frequency monitoring~\cite{Arcolezi2023evolving,Arcolezi2022,rappor,microsoft,Vidal2020}.

More precisely, LDP-Auditor relies on Monte Carlo methods to estimate the probabilities $\hat{p}_0=\Pr[\calM(v_1) \in O]$ and $\hat{p}_1=\Pr[\calM(v_2) \in O]$ from Equation~\eqref{eq:ldp} through attacks. 
From this, an empirical privacy loss is computed, $\epslb=\ln\left(\left(\hat{p}_0 - \delta\right)/\hat{p}_1\right)$, thus providing an estimate of the algorithm's privacy leakage. 
A comprehensive discussion on the LDP-Auditor framework, including its detailed methodology and applications, is deferred to Section~\ref{sec:ldp_auditing}.

Unlike traditional DP-SGD auditing, in which the focus is on distinguishing neighboring datasets, LDP-Auditor assesses the distinguishability of inputs directly. 
To achieve this, we instantiate LDP-Auditor with distinguishability attacks based on recent adversarial analysis of LDP frequency estimation protocols~\cite{Gursoy2022, Arcolezi2023}. 
These attacks allow an adversary's to predict the user's input value based on the obfuscated output, enabling LDP-Auditor to directly evaluate the privacy guarantees offered by LDP mechanisms, making it well-suited for privacy auditing.
In this context, expanding beyond~\cite{Gursoy2022, Arcolezi2023}, we also introduce novel distinguishability attacks tailored to four additional LDP frequency estimation protocols based on histogram encoding~\cite{tianhao2017}, as well as general distinguishability attacks on LDP protocols for longitudinal studies (see Algorithm~\ref{alg:attack_ldp_long}) and on LDP protocols for multidimensional data (see Algorithm~\ref{alg:attack_rs+fd}).

As an example, Figure~\ref{fig:summary_audit} illustrates an instance of our auditing results for a theoretical upper bound of $\epsilon=2$ (indicated by the dashed red line) across eight $\epsilon$-LDP frequency estimation protocols: Generalized Randomized Response (GRR)~\cite{kairouz2016discrete}, Subset Selection (SS)~\cite{wang2016mutual,Min2018}, Symmetric Unary Encoding (SUE)~\cite{rappor}, Optimal Unary Encoding (OUE)~\cite{tianhao2017}, Thresholding with Histogram Encoding (THE)~\cite{tianhao2017}, Summation with Histogram Encoding (SHE)~\cite{Dwork2006}, Binary Local Hashing (BLH)~\cite{Bassily2015} and Optimal Local Hashing (OLH)~\cite{tianhao2017}.
Among all these protocols, GRR demonstrated a tight empirical privacy loss estimation for $\epslb$ as it does not require a specific encoding. 
On the other hand, other LDP protocols presented $\epslb$ within $\leq 2$x of the theoretical $\epsilon$ (such as SUE, THE and SHE), and even within $\leq 4$x of the theoretical $\epsilon$ (like BLH).
\textit{These results indicate that either the state-of-the-art attacks are still not representative of the worst-case scenario or that the upper bound analyses of these LDP protocols are not tight.
The latter assumption might occur for LDP protocols that incorporate sources of randomness (\eg, due to hashing~\cite{Hadamard,Bassily2015,tianhao2017,apple})} not captured in the worst-case definition of LDP in Equation~\eqref{eq:ldp}.

More specifically, \textit{\textbf{we have investigated several factors influencing the audit}}, including the effect of theoretical privacy loss parameters ($\epsilon$ and $\delta$) in low, mid and high privacy regimes as well as the impact of the domain size $k$ on local privacy estimation. 
\textbf{\textit{Our investigation included detailed case studies to further explore specific facets of LDP auditing.}}
Notably, our analysis assessed how variations in $\delta$ affect the empirical privacy loss, $\epslb$, for approximate LDP variants~\cite{Wang2021_approx_ldp} of the GRR, SUE, BLH and OLH protocols, alongside with the Gaussian Mechanism (GM)~\cite{dwork2014algorithmic} and the Analytic GM (AGM)~\cite{balle18a}.
Moreover, given that BLH exhibited the least tight empirical privacy loss estimation $\epslb$, we investigated the privacy loss of local hashing without LDP obfuscation.
In addition, we examined the degradation of the empirical local privacy loss in repeated data collections compared to the theoretical upper bound imposed by the (L)DP sequential composition~\cite{dwork2014algorithmic}.
In this context, within a generic framework, we proposed distinguishability attacks on LDP protocols in \textit{longitudinal studies} (\cf{} Algorithm~\ref{alg:attack_ldp_long}).
Furthermore, we addressed the case of \textit{multidimensional data}, proposing distinguishability attacks for LDP protocols following the RS+FD~\cite{Arcolezi_rs_fd} solution (\cf{} Algorithm~\ref{alg:attack_rs+fd}).
We also show how \textit{LDP-Auditor successfully identified a bug in one state-of-the-art LDP Python package}, in which the empirical privacy loss $\epslb$ contradicts the theoretical upper bound $\epsilon$ (see Figure~\ref{fig:audit_pure_ldp}).

\begin{figure}[t]
    \centering
    \includegraphics[width=0.7\linewidth]{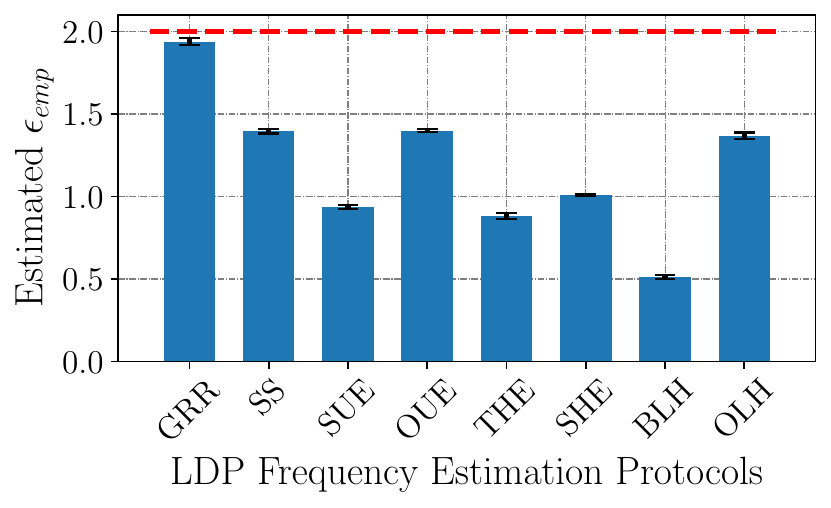}
    \caption{Comparison of estimated privacy loss $\epslb$ with theoretical upper bound $\epsilon=2$ for eight pure LDP frequency estimation protocols. 
    The dashed red line corresponds to the certifiable upper bound. 
    While GRR closely aligns with the theoretical bound, others exhibit empirical $\epslb$ within $\leq 2$x (\eg, SUE) or even $\leq 4$x (\ie, BLH) of the theoretical $\epsilon$ value.}
    \label{fig:summary_audit}
\end{figure}

Taking all these aspects into account, the coverage of our analysis is broadened, allowing for a more comprehensive assessment of the robustness of various LDP protocols in realistic data collection scenarios.
More specifically, our main contributions in this paper can be summarized as follows:

\begin{itemize}
    \item We introduce the LDP-Auditor framework, which aims to estimate the empirical privacy loss of LDP frequency estimation protocols.
     This framework provides a realistic assessment of privacy guarantees, which is essential for making informed decisions about LDP parameter selection and on stimulating the research of new privacy attacks. 

    \item We introduce novel distinguishability attacks specifically tailored to LDP protocols for longitudinal studies and multidimensional data. 
    These new attacks enrich the privacy analysis techniques available for examining the robustness of LDP mechanisms in practical settings.
    
    \item We conduct an extensive audit of various LDP protocols, analyzing the impact of factors such as privacy regimes, domain size and multiple data collections.
    This comprehensive analysis provides valuable insights into the resilience and effectiveness of nine state-of-the-art LDP mechanisms, fundamental building blocks for applications such as frequency monitoring~\cite{Arcolezi2023evolving,rappor,microsoft,Vidal2020}, heavy hitter estimation~\cite{Bassily2015,Wang2021} and machine learning~\cite{Chamikara2020,Yilmaz2020}.
    
    \item We demonstrate the bug detection capabilities of LDP-Auditor by identifying an issue in a state-of-the-art LDP Python package.
    This highlights the practical significance of our framework in validating LDP implementations.
\end{itemize}


\section{Related Work} 
\label{sec:rel_work}

Differential privacy auditing, as introduced by~\citet{Jagielski2020}, involves employing various techniques to empirically assess the extent of privacy leakage in machine learning algorithms through estimating the $\epslb$ privacy loss. 
These techniques are particularly valuable when known analytical bounds on the DP loss lack precision, allowing for empirical measurements of privacy in such cases.
For instance, DP auditing has been extensively investigated in evaluating the mathematical analysis for the well-known DP-SGD algorithm proposed by~\citet{abadi2016deep}. 
The research literature on DP-SGD auditing covers both centralized~\cite{Jagielski2020,Lu2022,steinke2023privacy,pillutla2023unleashing,nasr2023tight,nasr2021adversary,cebere2024tighter,tramer2022debugging} and federated~\cite{maddock2022canife,andrew2023one} learning settings.
Beyond privacy-preserving machine learning, privacy auditing has also been studied for standard DP algorithms~\cite{ding2018detecting,bichsel2021dpsniper,gorla2022possibility,Askin2022,Yun_Lu2022}.
For instance, some of these works consider a fully black-box scenario (\ie, unknown DP mechanism) with the goal of estimating the $\epsilon$-(L)DP guarantee provided~\cite{gorla2022possibility,Askin2022,Yun_Lu2022}.
Another line of research~\cite{ding2018detecting,bichsel2021dpsniper} has been tailored to identify errors in algorithm analysis or code implementations, especially when derived lower bounds contradict theoretical upper bounds.
While the works in~\cite{ding2018detecting,bichsel2021dpsniper} could also be used to certify the $\epsilon$-LDP guarantee through Monte Carlo estimations, our work considers realistic privacy attacks to LDP mechanisms to empirically estimate the privacy loss $\epslb$.
In other words, they would be able to answer ``\textit{is the claimed $\epsilon$-LDP correct in this code implementation?}'', whereas we alternatively answer ``\textit{is the claimed $\epsilon$-LDP worst-case guarantee tight under state-of-the-art attacks?}''.

This distinction highlights our emphasis on assessing the tightness of privacy guarantees under stringent adversarial conditions. 
Consequently, we envision our auditing analysis as an stimulus for advancing the current state-of-the-art in privacy attacks on LDP protocols and achieve tight empirical estimates for $\epslb$.
In this context, the existing literature on privacy attacks on LDP comprises several categories: 
(1) Distinguishability attacks~\cite{Gursoy2022,Arcolezi2023,Chatzikokolakis2023} (adopted in this work), which enable adversaries to predict the users' input based on the obfuscated outputs; 
(2) Pool inference attacks~\cite{Gadotti2022}, allowing adversaries to deduce a user's preferences or attributes from the aggregated data, such as inferring a user's preferred skin tone used in emojis; 
(3) Re-identification attacks~\cite{Murakami2021,Arcolezi2023}, aiming to uniquely identify a specific user within a larger population; 
and (iv) Attacks on iterative data collections~\cite{Arcolezi2023evolving,GURSOY2024}, which allows adversaries to detect a pattern change in longitudinal studies, such as when someone starts a diet by monitoring calorie consumption.

\section{LDP Frequency Estimation Protocols} 
\label{sec:background}

In this section, we review the necessary notation (\cf{} Table~\ref{tab:notation} in Appendix~\ref{app:notation}) and background information of the LDP frequency estimation protocols.
Throughout the paper, let $[n]=\{1, 2, \ldots, n\}$ denote a set of integers and $V = \{v_1, \ldots, v_k\}$ represent a sensitive attribute with a discrete domain of size $k = |V|$. 
We consider a distributed setting with $n$ users and one untrusted server collecting the data reported by these users.
The fundamental premise of ($\epsdelta$)-LDP, as stated in Equation~\eqref{eq:ldp}, is that the input to $\calM$ cannot be confidently determined from its output, with the level of confidence determined by $e^{\epsilon}$ and $\delta$. 
Therefore, the user's privacy is considered compromised if the adversary can correctly predict the user's value.

In recent works~\cite{Gursoy2022, Arcolezi2023}, the authors introduced \textbf{distinguishability attacks} $\calA$ to state-of-the-art LDP frequency estimation protocols.
These attacks enable an adversary to predict the users' value $\hat{v}=\calA(y)$, in which $y=\calM(v)$ represents the reported value obtained through the $\epsilon$-LDP protocol.
In essence, although each LDP protocol employs different encoding and perturbation functions, the adversary's objective remains the same, namely to predict the user's true value by identifying the most likely value that would have resulted in the reported value $y$.
The notion of distinguishability attacks provides a unified approach to evaluate the privacy guarantees offered by different LDP protocols.

We now provide a brief overview of state-of-the-art pure and approximate LDP frequency estimation protocols $\calM$, along with their respective distinguishability attacks denoted as $\calA_{\calM}$.
The attack $\calA_{\calM}$ generally relies on a ``support set''~\cite{tianhao2017}, denoted as $\setOne_{\calM}$, which is built upon the reported value $y$.
The combination of these protocols and attack strategies will enable us to comprehensively audit the empirical privacy level provided by various LDP mechanisms.

\subsection{Pure $\epsilon$-LDP Protocols} \label{sub:pure_ldp_protocols}

\textbf{Generalized Randomized Response (GRR).} The GRR~\cite{kairouz2016discrete} mechanism generalizes the randomized response surveying technique proposed by~\citet{Warner1965} for $k \geq 2$ while satisfying $\epsilon$-LDP.
Given a value $v \in V$, $\mathrm{GRR}(v)$ outputs the true value $v$ with probability $p$, and any other value $v' \in V \setminus \{v\}$, otherwise. 
More formally:
\begin{equation} \label{eq:grr}
    \Pr[\mathrm{GRR}(v)=y] = \begin{cases} p=\frac{e^{\epsilon}}{e^{\epsilon}+k-1} \textrm{ if } y = v,\\ q=\frac{1}{e^{\epsilon}+k-1} \textrm{ if } y \neq v \textrm{,} \end{cases}
\end{equation}

\noindent in which $y \in V$ is the perturbed value sent to the server. 
The support set for GRR is simply $\setOne_{\mathrm{GRR}}=\{y\}$.
From Equation~\eqref{eq:grr}, $\Pr[y=v] > \Pr[y=v']$ for all $v' \in V \setminus \{v\}$. 
Therefore, the attack strategy $\calA_{\mathrm{GRR}}$ is to predict $\hat{v}=y$~\cite{Gursoy2022,Arcolezi2023}.

\textbf{Subset Selection (SS).} The SS~\cite{wang2016mutual,Min2018} mechanism was proposed for the case in which the obfuscation output is a subset of values $\Omega$ of the original domain $V$.
The optimal subset size that minimizes the variance is $\omega= |\Omega| = \max \left (1, \left\lfloor \frac{k}{e^{\epsilon}+1} \right\rceil \right)$. 
Given an empty subset $\Omega$, the true value $v$ is added to $\Omega$ with probability $p=\frac{\omega e^{\epsilon}}{\omega e^{\epsilon} + k - \omega}$. 
Finally, values are added to $\Omega$ as follows:

\begin{itemize}
    \item If $v \in \Omega$, then $\omega - 1$ values are sampled from $V \setminus \{v\}$ uniformly at random (without replacement) and are added to $\Omega$;
    
    \item If $v \notin \Omega$, then $\omega$ values are sampled from $V \setminus \{v\}$ uniformly at random (without replacement) and are added to $\Omega$.
\end{itemize}

Afterward, the user sends the subset $\Omega$ to the server.
The support set for SS is the subset of all values in $\Omega$, \ie, $\setOne_{\mathrm{SS}}=\{v | v \in \Omega\}$.
Therefore, the attack strategy $\calA_{\mathrm{SS}}$ is to predict $\hat{v}=\mathrm{Uniform}\left( \setOne_{\mathrm{SS}} \right)$~\cite{Gursoy2022,Arcolezi2023}.

\textbf{Unary Encoding (UE).} UE protocols~\cite{rappor,tianhao2017} encode the user's input data $v \in V$, as a one-hot $k$-dimensional vector before obfuscating each bit independently.
More precisely, let $\textbf{v}=[0, \ldots, 0, 1, 0, \ldots, 0]$ be a binary vector with only the bit at the position $v$ set to $1$ while the other bits are set to $0$.
The obfuscation function of UE mechanisms randomizes the bits from $\textbf{v}$ independently to generate $\textbf{y}$ as follows:
\begin{equation}  \label{eq:ue_parameters}
    \forall{i \in [k]} : \quad \Pr[\textbf{y}_i=1] =\begin{cases} p, \textrm{ if } \textbf{v}_i=1 \textrm{,} \\ q, \textrm{ if } \textbf{v}_i=0 \textrm{,}\end{cases}
\end{equation}

\noindent in which $\textbf{y}$ is sent to the server. 
There are two variations of UE mechanisms: (i) Symmetric UE (SUE)~\cite{rappor} that selects $p=\frac{e^{\epsilon/2}}{e^{\epsilon/2}+1}$ and $q=\frac{1}{e^{\epsilon/2}+1}$ in Equation~\eqref{eq:ue_parameters}, such that $p+q=1$; and (ii) Optimal UE (OUE)~\cite{tianhao2017} that selects $p=\frac{1}{2}$ and $q=\frac{1}{e^{\epsilon}+1}$ in Equation~\eqref{eq:ue_parameters}.
With $\textbf{y}$, the adversary can construct the subset of all values $v \in V$ that are set to 1, \ie, $\setOne_{\mathrm{UE}}=\{v | \textbf{y}_v = 1\}$.
There are two possible attack strategies $\calA_{\mathrm{UE}}$~\cite{Gursoy2022,Arcolezi2023}:

\begin{itemize}
    \item $\calA^0_{\mathrm{UE}}$ is a random choice $\hat{v}=\mathrm{Uniform}\left( [k] \right)$, if $\setOne_{\mathrm{UE}}=\emptyset$;
    
    \item $\calA^1_{\mathrm{UE}}$ is a random choice $\hat{v}=\mathrm{Uniform}\left( \setOne_{\mathrm{UE}} \right)$, otherwise.
\end{itemize}

\textbf{Local Hashing (LH).} LH protocols~\cite{tianhao2017,Bassily2015} use hash functions to map the input data $v \in V$ to a new domain of size $g \geq 2$, and then apply GRR to the hashed value. 
Let $\mathscr{H}$ be a universal hash function family such that each hash function $\mathrm{H} \in \mathscr{H}$ hashes a value $v \in V$ into $[g]$ (\ie, $\mathrm{H} : V \rightarrow [g]$). 
There are two variations of LH mechanisms: (i) Binary LH (BLH)~\cite{Bassily2015} that just sets $g=2$, and (ii) Optimal LH (OLH)~\cite{tianhao2017} that selects $g=\lfloor e^{\epsilon} + 1 \rceil$.
Each user first selects a hash function $\mathrm{H} \in \mathscr{H}$ at random and obfuscates the hash value $h=\mathrm{H}(v)$ with GRR. 
In particular, the LH reporting mechanism is $\mathrm{LH}(v) \coloneqq \langle \mathrm{H}, \mathrm{GRR}(h) \rangle \textrm{,}$ in which $\mathrm{GRR}(h)$ is given in Equation~\eqref{eq:grr} while operating on the new domain $[g]$. 
Each user reports the hash function and obfuscated value $\langle \mathrm{H}, y \rangle$ to the server. 
With these elements, the adversary can construct the subset of all values $v \in V$ that hash to $y$, \ie, $\setOne_{\mathrm{LH}}= \{v | \mathrm{H}(v) = y\}$.
There are two possible attack strategies $\calA_{\mathrm{LH}}$~\cite{Gursoy2022,Arcolezi2023}:
\begin{itemize}
    \item $\calA^0_{\mathrm{LH}}$ is a random choice $\hat{v}=\mathrm{Uniform}\left( [k] \right)$, if $\setOne_{\mathrm{LH}}=\emptyset$;
    
    \item $\calA^1_{\mathrm{LH}}$ is a random choice $\hat{v}=\mathrm{Uniform}\left( \setOne_{\mathrm{LH}} \right)$, otherwise.
\end{itemize}

\textbf{Histogram Encoding (HE).} HE protocols~\cite{tianhao2017} encode the user value as a one-hot $k$-dimensional histogram, $\textbf{v}=[0.0, 0.0, \ldots, 1.0, 0.0, \ldots, 0.0]$ in which only the $v$-th component is $1.0$.
To satisfy $\epsilon$-LDP, $\mathrm{HE}(\textbf{v})$ perturbs each bit of $\textbf{v}$ independently using the Laplace mechanism~\cite{Dwork2006}. 
Two different input values $v_1,v_2 \in V$ will result in two vectors with L1 distance of $\Delta_1=2$. 
Thus, HE will output $\textbf{y}$ such that $\textbf{y}_i = \textbf{v}_i + \textrm{Lap}\left( \frac{\Delta_1}{\epsilon} \right)$.
\emph{In this paper, we propose distinguishability attacks on two pure $\epsilon$-LDP HE protocols:}

\begin{itemize}
\item \textbf{Summation with HE (SHE)}~\cite{Dwork2006}. With SHE, there is no post-processing of $\textbf{y}$.
    Instead of constructing a support set, we describe our attacking strategy to SHE as follows.
    Let $P_V(v)$ be the prior probability of input value $v$, and let $P_Y(\textbf{y}|v)$ be the likelihood of observing $\textbf{y}$ given the true input value $v$. 
    By the Bayes' theorem, the posterior probability of input value $v$ given the observed $\textbf{y}$ is:
    \begin{equation} \label{eq:prior_bayes}
    P_V(v|\textbf{y}) = \frac{P_Y(\textbf{y}|v)P_V(v)}{\sum_{i=1}^{k}P_Y(\textbf{y}|i)P_V(i)} \mathrm{.}
    \end{equation}
    
    We can compute the likelihood $P_Y(\textbf{y}|v)$ as follows. 
    For a given $v$, the corresponding one-hot encoded histogram is $\textbf{v}$. 
    The reported value $\textbf{y}$ is the sum of $\textbf{v}$ and noise from a Laplace distribution with scale $b=2/\epsilon$. 
    Therefore, the likelihood of observing $\mathbf{y}$ given $\mathbf{v}$ is:
    
    \begin{equation} \label{eq:likelihood_laplace}
    P_Y(\textbf{y}|\textbf{v}) = \frac{1}{(2b)^k} \exp\left(-\frac{|\textbf{y}-\textbf{v}|_1}{b}\right) \mathrm{,}
    \end{equation}
    
    \noindent in which $|\textbf{y}-\textbf{v}|_1$ is the L1 distance between $\textbf{y}$ and $\textbf{v}$.
    To perform the attack, we compute the posterior probability $P_V(v|\textbf{y})$ for each possible input value $v \in V$ and output the most probable input value. 
    In other words, given the reported $\textbf{y}$, our Bayes optimal attack $\calA_{SHE}$ outputs:
    
    \begin{equation} \label{eq:prediction_HE}
    \hat{v} = \arg\max_{v \in V} P_V(v|\textbf{y}) \mathrm{.}
    \end{equation}
    
    Note that this attack requires knowledge of the prior probability distribution $P_V(v)$. 
    If the prior is unknown (assumed in this paper), one can use a uniform prior.
    
    \item \textbf{Thresholding with HE (THE)}~\cite{tianhao2017}. 
    With THE, the server (or the user) can construct the support set as $\setOne_{\mathrm{THE}} = \{ v \hspace{0.1cm} | \hspace{0.1cm} \textbf{y}_v \hspace{0.1cm} >  \hspace{0.1cm}\theta\}$, \ie, each noise count whose value $> \theta$.
    The optimal threshold value for $\theta$ that minimizes the protocol's variance is within $(0.5, 1)$.
    With $\setOne_{\mathrm{THE}} = \{ v \hspace{0.1cm} | \hspace{0.1cm} \textbf{y}_v \hspace{0.1cm} >  \hspace{0.1cm}\theta\}$, we propose an adversary $\calA_{\mathrm{THE}}$ with two attack strategies:
    \begin{itemize}
        \item $\calA^0_{\mathrm{THE}}$ is a random choice $\hat{v}=\mathrm{Uniform}\left( [k] \right)$, if $\setOne_{\mathrm{THE}}=\emptyset$;
        
        \item $\calA^1_{\mathrm{THE}}$ is a random choice $\hat{v}=\mathrm{Uniform}\left( \setOne_{\mathrm{THE}} \right)$, otherwise.
    \end{itemize}
    
\end{itemize}

\subsection{Approximate ($\epsdelta$)-LDP Protocols} \label{sub:approximate_ldp_protocols}

In this section, we describe two ($\epsdelta$)-LDP protocols, which are based on the Gaussian mechanism~\cite{dwork2014algorithmic,balle18a}.
We defer the descriptions of approximate ($\epsdelta$)-LDP variants~\cite{Wang2021_approx_ldp} of GRR, SUE and LH protocols -- namely, Approximate GRR (AGRR), Approximate SUE (ASUE), Approximate LH (ALH) -- to Appendix~\ref{app:approximate_ldp_protocols_detailed}.

\textbf{HE with Gaussian Mechanism (HE-GM)~\cite{dwork2014algorithmic,balle18a}.} Similar to HE protocols of Section~\ref{sub:pure_ldp_protocols}, HE-GM protocols encode the user value as a one-hot $k$-dimensional histogram.
Then, $\textrm{HE-GM}(\textbf{v})$ perturbs each bit of $\textbf{v}$ independently using a Gaussian mechanism (GM)~\cite{dwork2014algorithmic,balle18a}.
Two different input values $v_1, v_2 \in V$ will result in two vectors with L2 distance of $\Delta_2=\sqrt{2}$.
Thus, HE-GM will output $\textbf{y}$ such that $\textbf{y}_i = \textbf{v}_i + \mathcal{N}\left( 0, \sigma^2 \right)$, in which $\sigma$ is determined by $\epsdelta, \textrm{ and } \Delta_2$.
When using the well-established GM for $\epsdelta \in (0, 1)$, $\sigma= \frac{\Delta_2}{\epsilon} \sqrt{2 \ln(1.25/\delta)}$~\cite{dwork2014algorithmic}.
In this paper, we also consider the Analytic GM (AGM)~\cite{balle18a}, which is an improved version of the GM~\cite{dwork2014algorithmic} and can be applied for any $\epsilon>0$.
The main difference between GM and AGM is the method to parameterize $\sigma$.
With AGM, $\sigma$ is calculated analytically as demonstrated in~\cite[Algorithm 1]{balle18a} and its implementation~\cite{balle_agm}.
Hereafter, we will specifically denote ``AGM'' and ``GM'' when referring to HE-GM instantiated with AGM and GM, respectively.

Building upon our distinguishability attack's description of the SHE protocol with Laplace noise, we extend the attack analysis to HE-GM protocols. 
The overall strategy, including the use of Bayes' theorem to compute posterior probabilities, remains consistent with our prior description in Section~\ref{sub:pure_ldp_protocols} (\emph{cf.} Equation~\eqref{eq:prior_bayes}). 
However, the key difference lies in the noise distribution used for ensuring LDP.

While the Laplace mechanism involves adding noise drawn from a Laplace distribution with scale $b=2/\epsilon$, the Gaussian mechanism adds noise following the normal distribution, (\ie, $\mathcal{N}\left( 0, \sigma^2 \right)$). 
This needs a different computation for the likelihood $P_Y(\mathbf{y}|v)$ in Equation~\eqref{eq:likelihood_laplace}, reflecting the properties of Gaussian noise.
Accordingly, the likelihood of observing $\mathbf{y}$ given $\mathbf{v}$ under Gaussian noise is:

\begin{equation} \label{eq:likelihood_gaussian}
    P_Y(\mathbf{y}|\mathbf{v}) = \frac{1}{\sqrt{(2\pi\sigma^2)^k}} \exp\left(-\frac{|\mathbf{y}-\mathbf{v}|_2^2}{2\sigma^2}\right) \mathrm{,}
\end{equation}

\noindent in which $|\mathbf{y}-\mathbf{v}|_2^2$ denotes the L2 squared distance between $\mathbf{y}$ and $\mathbf{v}$.

Then, our Bayes optimal attack for HE-GM protocols $\calA_{\textrm{HE-GM}}$ predicts the most probable input value, $\hat{v}$, given the reported $\mathbf{y}$, by following Equation~\eqref{eq:prediction_HE}.
Remark that Equation~\eqref{eq:likelihood_gaussian} is valid for both GM and AGM as a function of their respective noise scale $\sigma$.
Similar to the $\calA_{SHE}$ attack, if the prior probability distribution $P_V(v)$ is unknown, a uniform prior may be assumed for the analysis.

\section{LDP Auditing} \label{sec:ldp_auditing}

In this section, we introduce our LDP-Auditor framework (Section~\ref{sub:ldp_auditor}) and our distinguishability attacks considering multiple data collections (Section~\ref{sub:ldp_auditor_long} and Section~\ref{sub:ldp_auditor_multidimensional}).

\subsection{LDP-Auditor} \label{sub:ldp_auditor}

Our LDP-Auditor framework builds upon previous work on central DP auditing~\cite{Jagielski2020} with slight modifications tailored for LDP auditing.
This adaptation is necessary due to the intrinsic differences between the central DP and LDP models, primarily regarding the granularity of privacy and the nature of the data being protected.
Figure~\ref{fig:adversarial_privacy_game} in Appendix~\ref{app:adv_priv_game} compares the adversarial privacy game between central and local DP.
Unlike central DP, in which the adversary's objective is to distinguish between two ``neighboring datasets'', the LDP model shifts the focus towards distinguishing between individual ``inputs''.
We instantiate LDP-Auditor with distinguishability attacks to construct a robust test statistic for auditing LDP mechanisms. 
Specifically, we can formulate a distinguishability attack as a binary hypothesis testing problem: $\mathcal{H}$: ``$y$ comes from $v_1$''.
The attacker receives an output drawn from one of the two distributions $\calM(v_1)$ or $\calM(v_2)$ and has to infer whether the input was $v_1$ or not.
If the algorithm $\calM$ is $(\epsdelta)$-LDP, then no distinguishability attacks can be too accurate~\cite{Kairouz2015}.
Specifically, for any distinguishability attack $\calA$, we can statistically measure LDP re-writing Equation~\eqref{eq:ldp} as:

\begin{equation} \label{eq:ldp_audit}
    \underbrace{\Pr[\calA(\calM(\textcolor{teal}{v_1})) = \textcolor{teal}{v_1}]}_{\textcolor{teal}{\text{True Positive Rate (TPR)}}} \leq e^{\epsilon} \cdot 
    \underbrace{\Pr[\calA(\calM(\textcolor{red}{v_2})) = \textcolor{teal}{v_1}]}_{\textcolor{red}{\text{False Positive Rate (FPR)}}} +\textrm{ } \delta \textrm{.}
\end{equation}

If $\calM$ satisfies $(\epsdelta)$-LDP, then $\epsilon \geq \ln \left (\frac{\textcolor{teal}{\textrm{TPR}} - \delta}{\textcolor{red}{\textrm{FPR}}} \right)$.
In this formulation, the TPR is the probability that the attack correctly identifies $y$ as coming from $v_1$, and the FPR is the likelihood that $y$ is incorrectly attributed to $v_1$ when it comes from $v_2$.
Note that the $\delta$ term in Equation~\eqref{eq:ldp_audit} reflects the privacy budget from $\calM$ and is not an independent probability or error rate introduced by the distinguishability attack $\calA$.
However, a single run of a distinguishability attack is typically not sufficient to draw meaningful conclusions due to the inherent variability in the mechanism's outputs. 
Thus, to ensure the robustness of our empirical privacy loss estimation and account for statistical uncertainty, LDP-Auditor runs for multiple trials $T$ in order to compute the TPR and FPR from Equation~\eqref{eq:ldp_audit}.
Then, to affirm that our empirical privacy loss estimation is valid with a probability greater than $1 - \alpha$, we use Clopper-Pearson confidence intervals\footnote{We briefly describe the generic Clopper-Pearson method in Appendix~\ref{app:clopper_pearson}.}~\cite{clopper1934use} to establish a lower bound $\hat{p}_1$ for the FPR and an upper bound $\hat{p}_0$ for the TPR, each with a confidence of $1 - \alpha/2$. 
As a consequence, we can be confident that our empirical privacy loss estimation $\epslb = \ln \left (\frac{\hat{p}_0 - \delta}{\hat{p}_1} \right)$, holds with probability $1 - \alpha$.
This procedure is outlined in Algorithm~\ref{alg:ldp_auditor_lb}, and we prove its correctness in Theorem~\ref{theorem:ldp_auditor}.
The proof of Theorem~\ref{theorem:ldp_auditor} is deferred to Appendix~\ref{app:proof_theorem_1}.

\begin{algorithm}[!ht]
\caption{LDP-Auditor.}
\label{alg:ldp_auditor_lb}
\begin{algorithmic}[1]

\Statex \textbf{Input :} Theoretical $\epsilon$ and $\delta$, LDP protocol $\calM$, distinguishability attack $\calA$, values $v_1, v_2 \in V$, trial count $T$, confidence level $\alpha$. 
\Statex \textbf{Output :} Estimated privacy loss $\epslb$.

\State $\mathrm{TP}=0$, $\mathrm{FP} = 0$   \Comment{True Positive (TP) and False Positive (FP)}

\State \textbf{for} $i \in [T]$ \textbf{do}  

\State  \hskip1em \textbf{if} $\calA(\calM(v_1)) = v_1$ $\quad \mathrm{TP} = \mathrm{TP} + 1$

\State  \hskip1em \textbf{if} $\calA(\calM(v_2)) = v_1$ $\quad \mathrm{FP} = \mathrm{FP} + 1$

\State \textbf{end for}

\State $\hat{p}_0=\textrm{ClopperPearsonLower}(\mathrm{TP}, T, \alpha/2)$ 

\State $\hat{p}_1=\textrm{ClopperPearsonUpper}(\mathrm{FP}, T, \alpha/2)$ 

\Statex \textbf{return :} $\epslb=\ln((\hat{p}_0 - \delta) / \hat{p}_1)$ 
\end{algorithmic}
\end{algorithm}

\begin{theorem}[Correctness of LDP-Auditor]
\label{theorem:ldp_auditor} Given black-box access to an LDP mechanism $\calM$, and a distinguishability attack $\calA$, for any two distinct values $v_1, v_2$, a number of trials $T$, and a statistical confidence $\alpha$, if LDP-Auditor in Algorithm~\ref{alg:ldp_auditor_lb} returns $\epslb$, then, with probability $1-\alpha$, $\calM$ does not satisfy $(\epsilon',\delta)$-LDP for any $\epsilon' < \epslb$.
\end{theorem}

\textbf{Accounting for statistical uncertainty.}  We highlight that when we refer to $\epslb$ as an empirical lower bound with probability $1 - \alpha$, this naming is solely due to the inherent randomness of the Monte Carlo sampling process, without the need for any specific modeling or assumptions. 
By increasing the number of trials $T$, we can progressively enhance our confidence level towards 1. 
Furthermore, the decision to employ the Clopper-Pearson method stems from its relevance when an exact confidence interval is desired, in contrast to approximate methods (\eg, heuristic approaches). 
This approach enables a more reliable safeguard against underestimating privacy risks, and has been widely used in central DP audit research~\cite{Jagielski2020,nasr2021adversary,tramer2022debugging,bichsel2021dpsniper}. 
In this work, we utilize the Clopper-Pearson implementation provided by the \texttt{proportion\_confint} method in the Python package \texttt{statsmodels} (\url{https://pypi.org/project/statsmodels/}).

\textbf{Choice of parameters.} Given that LDP frequency estimation protocols usually distribute noise uniformly at random, the estimation of the empirical privacy loss $\epslb$ is not contingent upon selecting values $v_1$ and $v_2$ to represent a ``worst-case scenario'', unlike in central DP audit. 
Considering $V = \{1, 2, \ldots, k\}$, in this work, we set $v_1 = 1$ and $v_2 = 2$. 
The performed tests revealed no statistical difference when experiments were conducted with $v_1 = 1$ and a dynamic $v_2 = \textrm{Uniform}(2, k)$. 
Considering the experimental setup parameters, the number of trials $T$ and the confidence level $1-\alpha$ should be chosen to balance computational efficiency with the robustness of the empirical privacy loss estimation. 
Typically, a larger $T$ enhances the reliability of $\epslb$ estimates, while a smaller $\alpha$ increases the confidence in these estimates. 
In this work, we recommend selecting $T$ to be sufficiently large to ensure stable estimates across multiple experiments (\eg, we set $T=10^6$) and setting $\alpha$ to reflect a high confidence level, such as $0.05$ or $0.01$, to underpin the statistical significance of the empirical findings.

\textbf{Limits on the empirical privacy loss estimation.} 
The $\epslb$ reported by Algorithm~\ref{alg:ldp_auditor_lb} is upper bounded by the theoretical $\epsilon$ but also by an upper bound imposed by Monte Carlo estimation, which will be denoted by $\epsopt$ and depends on $\alpha$ and $T$.
For instance, let $\alpha=0.01$ to get a $99\%$-confidence bound and $T=10^4$ trials.
Even if we get perfect inference accuracy with $\mathrm{TP}=T$ and $\mathrm{FP}=0$, the Clopper-Pearson confidence interval would produce $\hat{p}_0=0.9994$ and $\hat{p}_1=0.0006$, which implies an empirical privacy loss of $\epslb=7.42$.
This means, with 99\% probability, the true $\epsilon$ is at least $7.42$, and $\epsopt(\alpha, T) = 7.42$.

\subsection{LDP-Auditor for Longitudinal Studies} \label{sub:ldp_auditor_long}

In practice, the server often needs to collect users’ data periodically throughout multiple data collections (\ie, \textit{longitudinal studies}). 
Nevertheless, in the worst-case, one known result in (L)DP is that \textbf{repeated data collections have a linear privacy loss due to the sequential composition}~\cite{dwork2014algorithmic}.
This occurs because attackers can exploit ``averaging attacks'' to distinguish the user's actual value from the added noise.
For this reason, well-known LDP mechanisms for longitudinal studies such as RAPPOR~\cite{rappor} (deployed in Google Chrome) and $d$BitFlipPM~\cite{microsoft} (deployed in Windows 10), were designed with a \textit{memoization-based} solution.
We discuss how to audit LDP mechanisms based on memoization in Appendix~\ref{app:memoization_ldp}.

Given $\tau$ data collections, we aim to audit the empirical privacy loss of LDP protocols in comparison to the upper bound $\tau \epsilon$-LDP imposed by the (L)DP sequential composition.
Our main motivation is to evaluate how tight the sequential composition is for LDP protocols.
Furthermore, this audit will provide insights into the privacy implications of real-world applications similar to those implemented by Apple~\cite{apple}, in which memoization was not employed. 

In Algorithm~\ref{alg:attack_ldp_long}, we present the extension of distinguishability attacks on LDP protocols to longitudinal studies $\calA^{L}$.
In this context, the adversary's objective remains the same: to predict the user's true value by determining the most probable value that would have generated the reported value $y^t$ after $\tau$ data collections.
Notably, the adversary now possesses an increased knowledge due to random fresh noise being added to the user's value $v$ over $\tau$ times.
To perform the ``averaging attack'', in each data collection, the adversary constructs the ``support set'' based on the reported value $y^t$ and LDP mechanism $\calM$. 
The support set is then used to increment the knowledge (\ie, count) about the user's true value and what constitutes noisy data, ultimately predicting $\hat{v}$.
We highlight that the exceptions are HE-based protocols in which the notion of a support set is not applicable, namely SHE, GM and AGM, rendering Algorithm~\ref{alg:attack_ldp_long} inapplicable.
In these protocols, Laplace or Gaussian noise with a mean of $0$ is added in each data collection.
Consequently, the ``averaging attack'' is straightforward as it involves determining $\hat{v}$ by taking the \texttt{argmax} of the summation of all reports.
Formally, this is expressed as $\hat{v} = \mathrm{argmax}\left( \sum_{t=1}^{\tau} \mathbf{y}^t \right)$.

Finally, our LDP-Auditor framework (Algorithm~\ref{alg:ldp_auditor_lb}) can be used to estimate the privacy loss of LDP protocols in longitudinal studies. 
To achieve this, one can simply replace ``$\calA(\calM(v))$'' in Lines 3 and 4 of Algorithm~\ref{alg:ldp_auditor_lb} with ``$\calA^{L}(v)$'', \ie, the distinguishability attack outlined in Algorithm~\ref{alg:attack_ldp_long}, which already takes into account $\calM$.

\begin{algorithm}[t]
\caption{Distinguishability Attack in Longitudinal Study: $\calA^{L}$.}
\label{alg:attack_ldp_long}
\begin{algorithmic}[1]

\Statex \textbf{Input :} User value $v$, privacy guarantee $\epsilon$, LDP protocol $\calM$, number of data collections $\tau$. 
\Statex \textbf{Output :} Predicted value $\hat{v}$.

\State Initialize a $k$ sized zero-vector $\mathbf{z}=[0,0,\ldots,0]$

\State \textbf{for} $t \in [\tau]$ \textbf{do:} 

\State  \hskip1em User-side randomization $y^t=\calM(v)$

\State  \hskip1em Given $y^t$, adversary construct support set $\setOne_{\calM}$

\State  \hskip1em \textbf{for} $v \in \setOne_{\calM}$ \textbf{do:}

\State  \hskip2em Increment count $\textbf{z}[v] = \textbf{z}[v] + 1$

\State \hskip1em \textbf{end for}

\State \textbf{end for}

\State Predict $\hat{v}=\mathrm{argmax}(\textbf{z})$

\Statex \textbf{return :} $\hat{v}$
\end{algorithmic}
\end{algorithm}

\subsection{LDP-Auditor for Multidimensional Data} \label{sub:ldp_auditor_multidimensional}

Another dimension of interest to the server is \textit{multidimensional data} (\ie, $d \geq 2$ attributes), aiming to enable more comprehensive decision-making.
Considering potential correlations among these attributes, the principles of DP sequential composition~\cite{dwork2014algorithmic} remain applicable in this context.
Therefore, the existing solutions for multidimensional data, represented as $\mathbf{v}=[v_1,v_2,\ldots,v_d]$, include:

\begin{itemize}
    \item \textbf{Splitting (SPL):} This naïve method involves partitioning the privacy budget $\epsilon$ among the $d$ attributes, collecting each attribute under $\frac{\epsilon}{d}$-LDP. 
    Examples based on this SPL solution are the LoPub~\cite{Ren2018} and Castell~\cite{kikuchi2022castell} mechanisms, which are designed for joint distribution estimation.
    
    \item \textbf{Sampling (SMP):} In this approach, users are divided into $d$ disjoint sub-groups. 
    Each sub-group $j \in [d]$ then reports the $j$-th attribute under $\epsilon$-LDP. 
    Example of mechanisms using the SMP solution include CALM~\cite{Zhang2018} and FELIP~\cite{Filho2023}, proposed for marginal estimation, and~\cite{nguyen2016collecting,wang2019,Wang2021_approx_ldp}, which introduced LDP mechanisms for mean estimation.
    
    \item \textbf{Random Sampling Plus Fake Data (RS+FD)~\cite{Arcolezi_rs_fd}:} In this solution, each user samples a single attribute $j \in [d]$ to report $v_j$ under $\epsilon'$-LDP and reports uniform fake data for the $d-1$ non-sampled attributes.
    Because the sampling result is not disclosed to the aggregator, there is amplification by sampling~\cite{Li2012,Balle2018}.
    For this reason, RS+FD utilizes an amplified privacy budget $\epsilon'=\ln{\left( d \cdot (e^{\epsilon} - 1) + 1 \right)}$ for the sampled attribute.
    An example based on RS+FD is the GRR-FS mechanism~\cite{bhaila2023local}, designed for node-level LDP on graph data, to enable training of graph neural networks.
\end{itemize}

Upon closer examination of the three solutions, one can notice that both SPL and SMP solutions can be considered as straightforward instances of reporting one attribute with a given LDP mechanism (one at a time for SPL).
Consequently, our LDP-Auditor framework can be directly used to estimate empirical privacy losses $\epslb$ for LDP mechanisms following the SPL and SMP solutions.
Therefore, in this work, our focus shifts towards auditing the RS+FD solution, for which there is a privacy amplification effect due to uncertainty on the server side.

In Algorithm~\ref{alg:attack_rs+fd}, we introduce the distinguishability attack designed for LDP protocols following the RS+FD solution, denoted as $\calA^{\textrm{RS+FD}}$.
Here, the adversary's objective is twofold: first, to predict the attribute that the user has sampled, and subsequently, to predict the user's actual value.
Since each user selects an attribute $j \in [d]$ uniformly at random, the Bayes optimal guess for the adversary is $\hat{j} = \mathrm{Uniform}([d])$.
Once the attribute is predicted, the adversary constructs the ``support set'' based on the reported value $y_{\hat{j}}$ and LDP mechanism $\calM$.
With the support set, as in Section~\ref{sec:background}, the adversary predicts the user's value $\hat{v}_{\hat{j}}$.

Finally, we extend our LDP-Auditor framework for RS+FD protocols in Algorithm~\ref{alg:ldp_auditor_rs+fd}. 
The main change is due to the multidimensional data setting, for which we define $\mathbf{v_1}$ and $\mathbf{v_2}$ in Lines 1 and 2 of Algorithm~\ref{alg:ldp_auditor_rs+fd}.
The test statistic remains unchanged, as it is derived from distinguishability attacks as per Algorithm~\ref{alg:attack_rs+fd}.
Notice that one main difference with RS+FD auditing is that even if the user did not sample the attribute $\hat{j}$, the attack can still predict the user's value $v_j$ correctly due to uniform fake data generation for that attribute.
Our goal is thus to audit if RS+FD satisfies the claimed $\epsilon$-LDP guarantee with amplification by sampling.
Algorithm~\ref{alg:ldp_auditor_rs+fd} builds upon the foundational principles established in Section~\ref{sub:ldp_auditor} and in Algorithm~\ref{alg:ldp_auditor_lb} and, consequently, the framework's correctness and reliability extend to this adaptation as well.

\begin{algorithm}[htb]
\caption{Distinguishability Attack on RS+FD: $\calA^{\textrm{RS+FD}}$.}
\label{alg:attack_rs+fd}
\begin{algorithmic}[1]

\Statex \textbf{Input :} User values $\mathbf{v}=[v_1,v_2,\ldots,v_d]$, privacy guarantee $\epsilon$, RS+FD protocol $\calM$. 
\Statex \textbf{Output :} Predicted value $\hat{v}_{\hat{j}}$.

\State \textbf{for} $i \in [d]$ \textbf{do:}  \Comment{\cf{} RS+FD~\cite[Algorithm 1]{Arcolezi_rs_fd}}

\State  \hskip1em User-side randomization $y_i=\calM(v_i)$

\State \textbf{end for}

\State Adversary predict user's sampled attribute $\hat{j} = \mathrm{Uniform}([d])$

\State Given $y_{\hat{j}}$, construct support set $\setOne_{\calM}$

\State Predict $\hat{v}_{\hat{j}}=\mathrm{Uniform}(\setOne_{\calM})$ \Comment{\cf{} Section~\ref{sec:background}}

\Statex \textbf{return :} $\hat{v}_{\hat{j}}$
\end{algorithmic}
\end{algorithm}

\begin{algorithm}[!ht]
\caption{LDP-Auditor for RS+FD Protocols.}
\label{alg:ldp_auditor_rs+fd}
\begin{algorithmic}[1]

\Statex \textbf{Input :} Theoretical $\epsilon$, LDP protocol $\calM$, distinguishability attack $\calA^{\textrm{RS+FD}}$, values $v_1, v_2 \in V$, trial count $T$, confidence level $\alpha$. 
\Statex \textbf{Output :} Estimated privacy loss $\epslb$.

\State $\mathbf{v_1}=[v_1, v_1, \ldots, v_1]_{1\times d}$, $\mathbf{v_2}=[v_2, v_2, \ldots, v_2]_{1\times d}$

\State $\mathrm{TP}=0$, $\mathrm{FP} = 0$  \Comment{True Positive (TP) and False Positive (FP)}

\State \textbf{for} $i \in [T]$ \textbf{do}  

\State  \hskip1em \textbf{if} $\calA^{\textrm{RS+FD}}(\mathbf{v_1}) = v_1$ $\quad \mathrm{TP} = \mathrm{TP} + 1$

\State  \hskip1em \textbf{if} $\calA^{\textrm{RS+FD}}(\mathbf{v_2}) = v_1$ $\quad \mathrm{FP} = \mathrm{FP} + 1$

\State \textbf{end for}

\State $\hat{p}_0=\textrm{ClopperPearsonLower}(\mathrm{TP}, T, \alpha/2)$ 

\State $\hat{p}_1=\textrm{ClopperPearsonUpper}(\mathrm{FP}, T, \alpha/2)$ 

\Statex \textbf{return :} $\epslb=\ln(\hat{p}_0 / \hat{p}_1)$ 
\end{algorithmic}
\end{algorithm}

\section{Experimental Evaluation} 
\label{sec:results}
This section presents our experimental setting to assess the proposed audit framework as well as the main results obtained.

\subsection{General Setup of Experiments} \label{sub:setup_experiments}

For all experiments, we have used the following setting:

\begin{itemize}
    \item \textbf{Environment.} All algorithms are implemented in Python 3 with the Numpy~\cite{numpy}, Numba~\cite{numba}, Ray~\cite{ray}, Multi-Freq-LDPy~\cite{multi_freq_ldpy} and pure-LDP~\cite{pureldp,Cormode2021} libraries, and run on a local machine with 2.50GHz Intel Core i9 and 64GB RAM. 
    Our LDP-Auditor tool is open-sourced in a GitHub repository~\cite{artifact_ldp_audit}.

    \item \textbf{Audit parameters.} We set $T=10^6$ trial counts and use Clopper-Pearson confidence intervals with $\alpha=0.01$ (\ie, our estimates hold with $99\%$ confidence).
    These parameters establish the Monte Carlo upper bound as $\epsopt=12.025$.

    \item \textbf{Stability.} Since LDP protocols are randomized, we report average results with standard deviation over 5 runs.
\end{itemize}

\subsection{Main Auditing Results} \label{sub:main_results}

We begin by presenting our main LDP auditing results, considering:

\begin{itemize}
    \item \textbf{LDP protocols.} We audit the eight $\epsilon$-LDP frequency estimation protocols described in Section~\ref{sub:pure_ldp_protocols} and the six ($\epsdelta$)-LDP frequency estimation protocols described in Section~\ref{sub:approximate_ldp_protocols}.

    \item \textbf{Theoretical upper bound.} We evaluated the LDP frequency estimation protocols in high, mid and low privacy regimes over the range $\epsilon \in \{0.25, 0.5, 0.75, 1, 2, 4, 6, 10\}$.
    The chosen range for $\epsilon$ follows the state-of-the-art LDP literature (\eg, see~\cite{Hadamard,Zhang2018,Chamikara2020,Wang2021,Wu2023}) and real-world implementations~\cite{desfontaines2021list} (\eg, RAPPOR~\cite{rappor} with $\epsilon=0.5$).

    \item \textbf{Delta parameter.} For approximate LDP, we set $\delta=1e^{-5}$.

    \item \textbf{Domain size.} We also varied the domain size $k \in \{25, 50, 100, 150, 200\}$ as it influences the performance of the distinguishability attacks.
\end{itemize}

Figure~\ref{fig:audit_pure_ldp_protocols} illustrates the theoretical $\epsilon$ values (x-axis) versus the estimated $\epslb$ values (y-axis), demonstrating the comparison across various domain sizes $k$, for the eight $\epsilon$-LDP frequency estimation protocols: GRR, SS, SUE, OUE, BLH, OLH, SHE and THE.
Similarly, Figure~\ref{fig:audit_approx_ldp_protocols} presents analogous plots for the six ($\epsdelta$)-LDP frequency estimation protocols: AGRR, ASUE, ABLH, AOLH, GM and AGM.
\textit{Henceforth, when discussing our results, the notation ``(A)GRR'' will be used whenever the findings are applicable to both GRR and AGRR protocols (analogously for other LDP protocols).}

\begin{figure*}[!ht]
    \centering
    \includegraphics[width=0.7\linewidth]{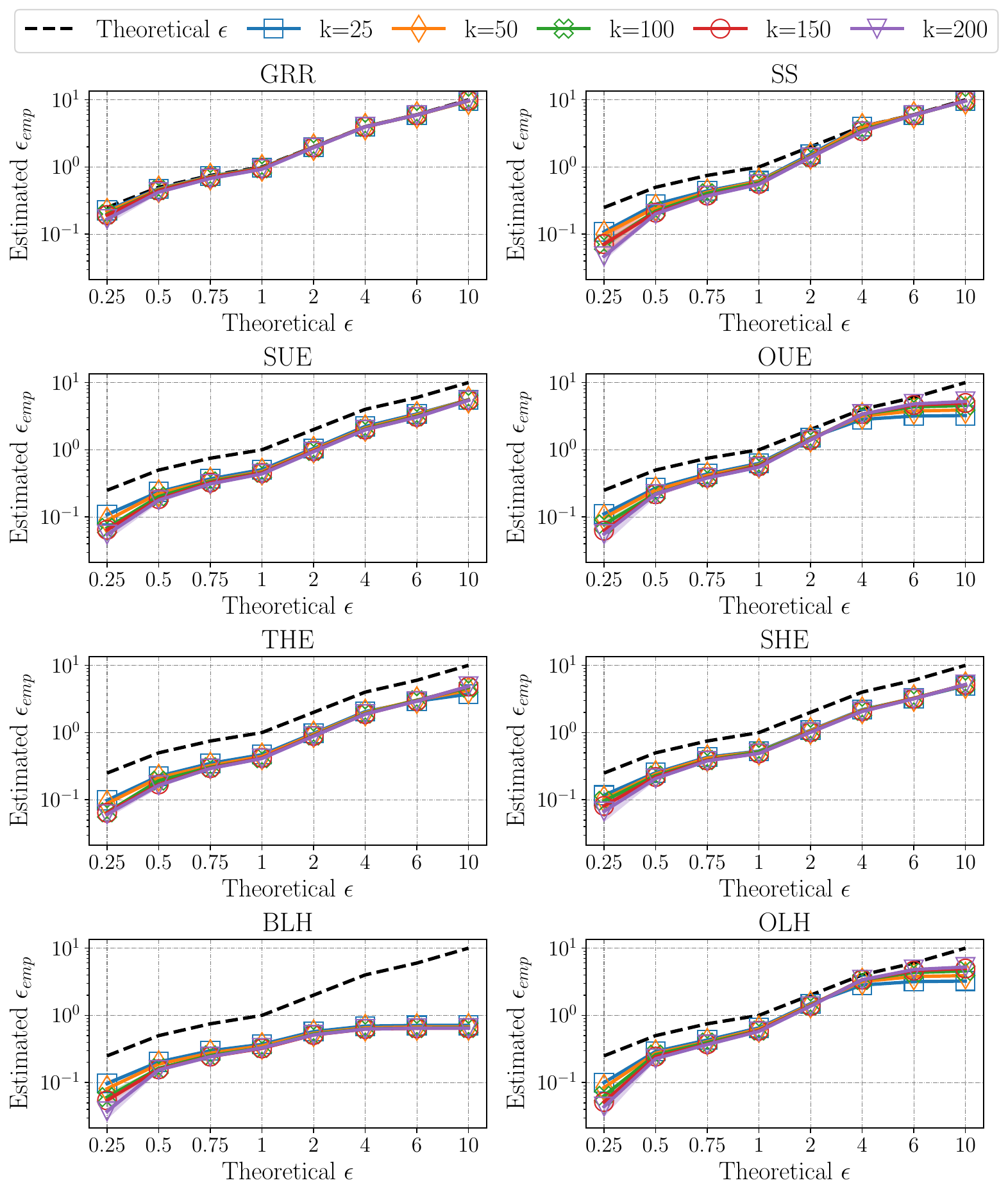}
    \caption{Theoretical $\epsilon$ values (x-axis) versus estimated $\epslb$ values (y-axis) using our LDP-Auditor framework with $\delta=0$. 
    We compare different domain sizes $k$ for eight state-of-the-art $\epsilon$-LDP frequency estimation protocols: GRR~\cite{kairouz2016discrete}, SS~\cite{wang2016mutual,Min2018}, SUE~\cite{rappor}, OUE~\cite{tianhao2017}, BLH~\cite{Bassily2015}, OLH~\cite{tianhao2017}, SHE~\cite{Dwork2006} and THE~\cite{tianhao2017}.}
    \label{fig:audit_pure_ldp_protocols}
\end{figure*}

\begin{figure*}[!ht]
    \centering
    \includegraphics[width=0.7\linewidth]{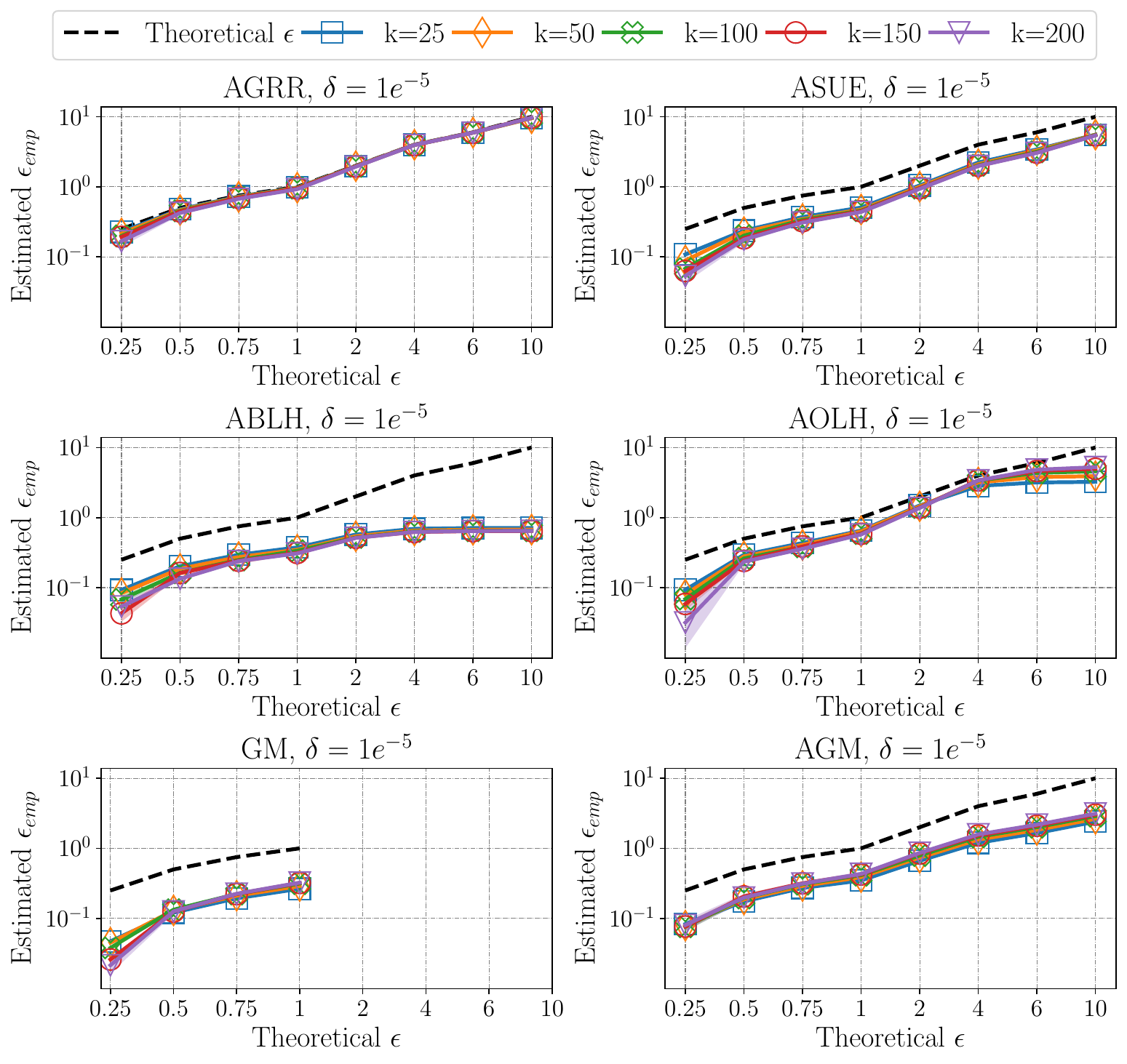}
    \caption{Theoretical $\epsilon$ values (x-axis) versus estimated $\epslb$ values (y-axis) using our LDP-Auditor framework with $\delta=1e^{-5}$. 
    We compare different domain sizes $k$ for six state-of-the-art ($\epsdelta$)-LDP frequency estimation protocols: AGRR~\cite{Wang2021_approx_ldp}, ASUE~\cite{Wang2021_approx_ldp}, ABLH~\cite{Wang2021_approx_ldp}, AOLH~\cite{Wang2021_approx_ldp}, GM~\cite{dwork2014algorithmic} and AGM~\cite{balle18a}.
    For GM, we only audit for certifiable theoretical upper bounds $\epsilon \leq 1$.}
    \label{fig:audit_approx_ldp_protocols}
\end{figure*}

\textbf{Effect of Encoding and Perturbation Functions.} It is important to note that LDP frequency estimation protocols employ different encoding and perturbation functions, leading to varying levels of susceptibility to distinguishability attacks~\cite{Gursoy2022,Arcolezi2023}.
Notably, as shown in Figure~\ref{fig:audit_pure_ldp_protocols} and Figure~\ref{fig:audit_approx_ldp_protocols}, one can notice that \textit{(A)GRR is the unique LDP protocol that achieves tight empirical privacy estimates for $\epslb$}.
As described in Section~\ref{sub:pure_ldp_protocols}, auditing (A)GRR's privacy guarantees is straightforward since there is no specific encoding (\ie, the input and output spaces are equal).
Conversely, all other LDP protocols (\ie, SS, UE-, LH- and HE-based) incorporate specific pre-processing encoding functions, which may result in information loss and/or additional randomness.

For instance, (A)BLH hashes the input set $V$ of size $k$ to $\{0,1\}$ and, thus results in excessive loss of information due to collisions.
Even if the bit is transmitted correctly after the (A)GRR perturbation, the server can only obtain one bit of information about the input (\ie, to which half of the input domain the value belongs to).
For these reasons, \textit{BLH consistently led to the worst auditing results} among the $\epsilon$-LDP protocols with a ``flat'' $\epslb<1$ estimation after $\epsilon \geq 2$.
Indeed, although both (A)LH protocols present similar empirical privacy losses $\epslb$ in high privacy regimes (\textit{the lowest among all other LDP protocols}), the difference is remarkable in favor of (A)OLH in mid to low privacy regimes. 
Thus, \textit{(A)OLH preserves more utility than (A)BLH, while providing tighter privacy loss estimation}.

Concerning the SS protocol that reports a subset $\Omega$ of $\omega$ values, one can note from Figure~\ref{fig:audit_pure_ldp_protocols} that the empirical privacy loss $\epslb$ demonstrated similar results to other LDP protocols in high privacy regimes. 
However, an exception occurs in low privacy regimes, in which SS equals GRR due to a subset size $\omega=1$, resulting in tight estimates for $\epslb$.
Regarding UE-based protocols, in high-privacy regimes ($\epsilon \leq 1$), both SUE and OUE presented similar empirical privacy estimates for $\epslb$ in Figure~\ref{fig:audit_pure_ldp_protocols}.
In mid-privacy regimes ($1 < \epsilon \leq 4$), OUE presented higher empirical privacy losses $\epslb$ than SUE. 
However, OUE reached a ``plateau'' estimation for $\epslb$ in low privacy regimes ($\epsilon > 4$), explained by an upper bound on the distinguishability attack (see~\cite{Gursoy2022}).
This plateau behavior is also observed for the (A)OLH protocol in low privacy regimes due to a comparable upper bound on the attacker effectiveness. 
Comparing approximate- and pure-SUE protocols, similar results were noticed for (A)SUE in Figure~\ref{fig:audit_pure_ldp_protocols} and Figure~\ref{fig:audit_approx_ldp_protocols}, considering all privacy regimes.

Lastly, for HE-based protocols, similar estimates for $\epslb$ were observed across all privacy regimes for both $\epsilon$-LDP protocols, namely SHE and THE, in Figure~\ref{fig:audit_pure_ldp_protocols}, albeit with varying sensitivity to the domain size $k$ (discussed afterwards). 
In contrast, from Figure~\ref{fig:audit_approx_ldp_protocols}, one can notice that the ($\epsdelta$)-LDP GM protocol led to the worst auditing results among all LDP protocols.
Therefore, \textit{in addition to AGM's ability to preserve greater utility than GM, it also offers more precise empirical privacy loss estimations}.

\textbf{Impact of domain size.} As the domain size $k$ increases, one can observe in Figure~\ref{fig:audit_pure_ldp_protocols} and Figure~\ref{fig:audit_approx_ldp_protocols} a direct impact on the empirical privacy loss estimation of $\epslb$ for all LDP protocols, in which the gap with the theoretical $\epsilon$ increases.
However, the impact is minor for the (A)GRR protocol, even in high privacy regimes.
Conversely, for all other LDP protocols, this impact is substantial, with empirical $\epslb$ estimates ranging within $\leq2.5$x of the theoretical $\epsilon$ (when $k=25$) up to $\leq5$x (when $k=200$).
These results are consistent with the distinguishability attack effectiveness, which decreases according to higher $k$ (\ie, more uncertainty)~\cite{Gursoy2022,Arcolezi2023}.
For instance, in the case of GRR, the probability $p=\frac{e^{\epsilon}}{e^{\epsilon}+k-1}$ of being ``honest'' in Equation~\eqref{eq:grr} decreases proportionally to $k$. 
In other mechanisms, there is a higher likelihood of introducing noise in the output $y$, such as by flipping more bits from 0 to 1 in (A)UE protocols.

Nevertheless, exceptions exist for both OUE and OLH protocols, in which in low privacy regimes (when $\epsilon \geq 4$), a larger domain size $k$ leads to tighter estimates of $\epslb$ than smaller domain sizes.
Although to a small extent, the THE protocol also yields more accurate estimates for higher $k$ when $\epsilon=10$.
Taking OUE as an example, these results can be attributed to the fact that the bit corresponding to the user's value is transmitted with a random probability of $\frac{1}{2}$ (\cf{} Equation~\eqref{eq:ue_parameters}). 
Consequently, if the domain size is small, it results in a higher false positive rate, which subsequently decreases the estimated empirical privacy loss $\epslb$.

\textbf{Generality of Our Findings.}
Overall, the gap between empirical $\epslb$ and theoretical $\epsilon$ privacy guarantees tends to widen in high privacy regimes (\ie, lower $\epsilon$ values).
This trend is particularly pronounced when considering the sensitivity of different LDP protocols to the domain size.
Lastly, we highlight that all $\epsilon$-LDP and ($\epsdelta$)-LDP frequency estimation protocols audited herein are building blocks of LDP mechanisms for more complex tasks such as:
heavy hitter estimation~\cite{Bassily2015,Wang2021}, joint distribution estimation~\cite{Ren2018,Zhang2018,kikuchi2022castell,Filho2023}, frequent item-set mining~\cite{Wang2018,Wu2023}, machine learning~\cite{Chamikara2020,Yilmaz2020}, frequency estimation of multidimensional data~\cite{nguyen2016collecting,wang2019,Arcolezi_rs_fd} and frequency monitoring~\cite{rappor,microsoft,Vidal2020,Arcolezi2022,Arcolezi2023evolving}.
Thus, our audit results provide generic insights that shed light on several critical factors influencing the estimation of the local privacy loss.

\subsection{Case Study \#1: Approximate- \emph{VS} Pure-LDP} \label{sub:audit_delta}

In theory,~\citet{Bun2019} proved that in the local DP model, approximate privacy is actually never more useful than pure privacy.
We will now compare approximate- and pure-LDP by assessing the impact of $\delta$ on the LDP auditing process.
In these experiments, we use the following parameter values:

\begin{itemize}
    \item \textbf{LDP protocols.} We audit the six ($\epsdelta$)-LDP protocols described in Section~\ref{sub:approximate_ldp_protocols}.
    
    \item \textbf{Theoretical upper bound.} Because GM requires $\epsilon \leq 1$~\cite{dwork2014algorithmic}, we vary the privacy guarantee only in high privacy regimes, within the range $\epsilon \in \{0.25, 0.5, 0.75, 1\}$.
    
    \item \textbf{Delta parameter.} We vary the $\delta$ parameter within the range $\delta \in \{0, 1e^{-7}, 1e^{-6}, 1e^{-5}, 1e^{-4}\}$; $\delta=0$ means $\epsilon$-LDP.

    \item \textbf{Domain size.} We vary the domain size $k \in \{25, 100, 150, 200\}$. 
    We present results for $k \in \{25, 200\}$ in the main paper and defer the others to Appendix~\ref{app:add_exp_impact_delta}
\end{itemize}

Figure~\ref{fig:audit_delta} illustrates the theoretical $\epsilon$ values (x-axis) versus estimated $\epslb$ values (y-axis) when varying the $\delta$ parameter and domain size $k \in \{25, 200\}$, using our LDP-Auditor framework. 
Note that for both GM and AGM protocols, there is no $\epslb$ value when $\delta = 0$, as these protocols do not have pure $\epsilon$-LDP variations.

Interestingly, for protocols such as AGRR, ASUE, ABLH and AOLH, our observations corroborate the theoretical assertions made by \citet{Bun2019} regarding the comparative utility of approximate versus pure privacy. 
More precisely, Figure~\ref{fig:audit_delta} and Figure~\ref{fig:appendix_audit_delta} (for $k \in \{100, 150\}$) indicate that variations in $\delta$ do not significantly alter the estimated privacy loss $\epslb$ across these protocols. 
This consistency in $\epslb$ values, irrespective of $\delta$ adjustments, suggests that for LDP protocols with a finite range, the audit outcomes for approximate-privacy closely align with those for pure-privacy. 
Conversely, the GM and AGM protocols exhibit distinct behaviours. 
More precisely, as $\delta$ increases, signaling a relaxation in the privacy constraint, we observe a narrowing gap between theoretical $\epsilon$ and empirical $\epslb$ values. 
This trend highlights a crucial aspect of LDP protocols with an infinite range, in which allowing for a nonzero $\delta$ directly influences the perceived privacy protection, leading to a more pronounced estimation of the privacy loss.
Finally, the impact of the domain size on the estimated privacy loss has a minor effect on the AGRR, ASUE, ABLH and AOLH protocols, with a decreasing $\epslb$ value for higher $k$. 
In contrast, for both GM and AGM protocols, the estimated $\epslb$ values increase (\ie, indicating less privacy) as $k$ increases.

\begin{figure}[!htb]
    \centering
    \begin{subfigure}{\columnwidth}
        \includegraphics[width=1.0\linewidth]{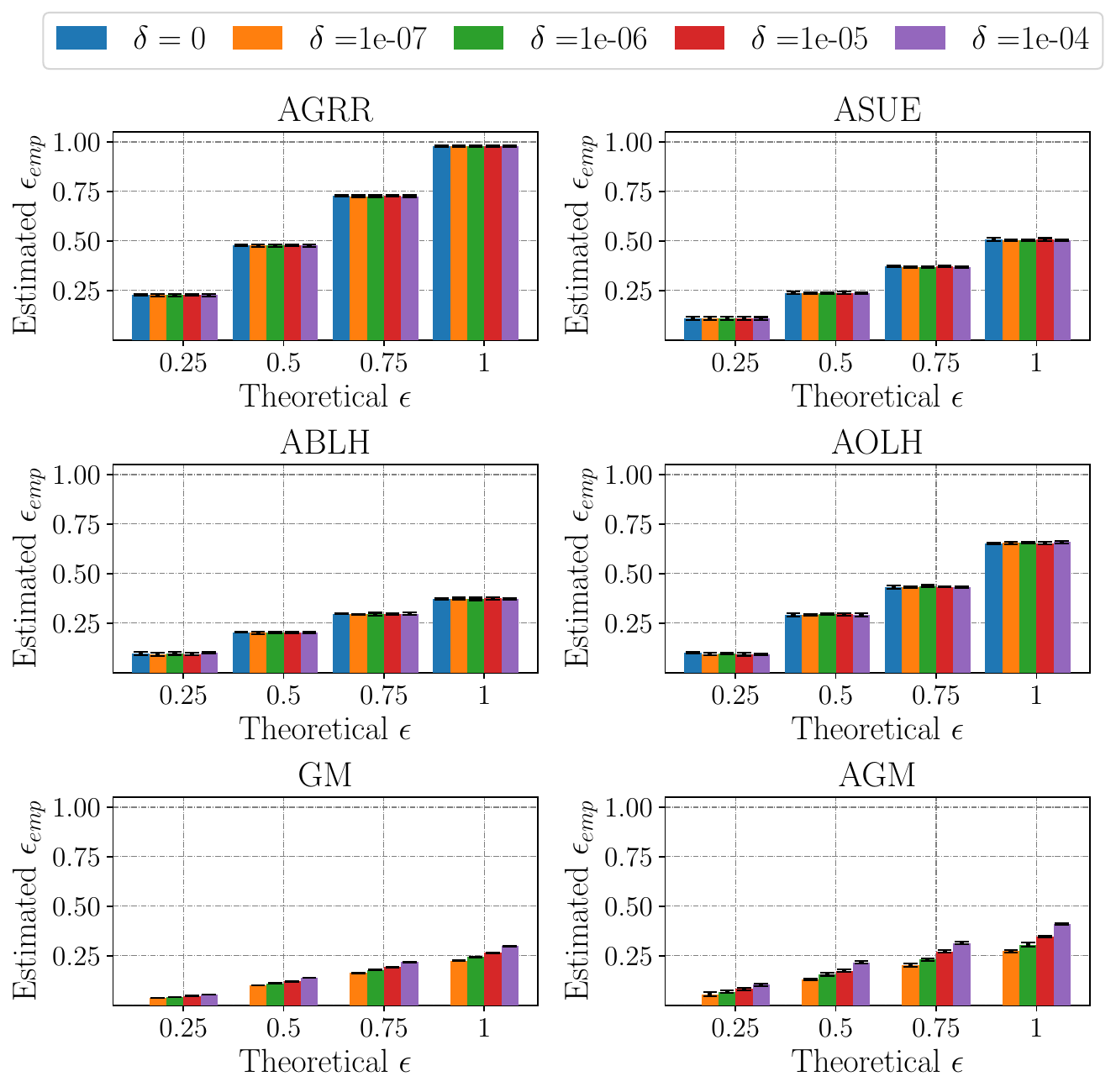}
        \caption{Domain size $k=25$.}
        \label{subfig:audit_delta_k25}
    \end{subfigure}
    \hfill
    \begin{subfigure}{\columnwidth}
        \includegraphics[width=1.0\linewidth]{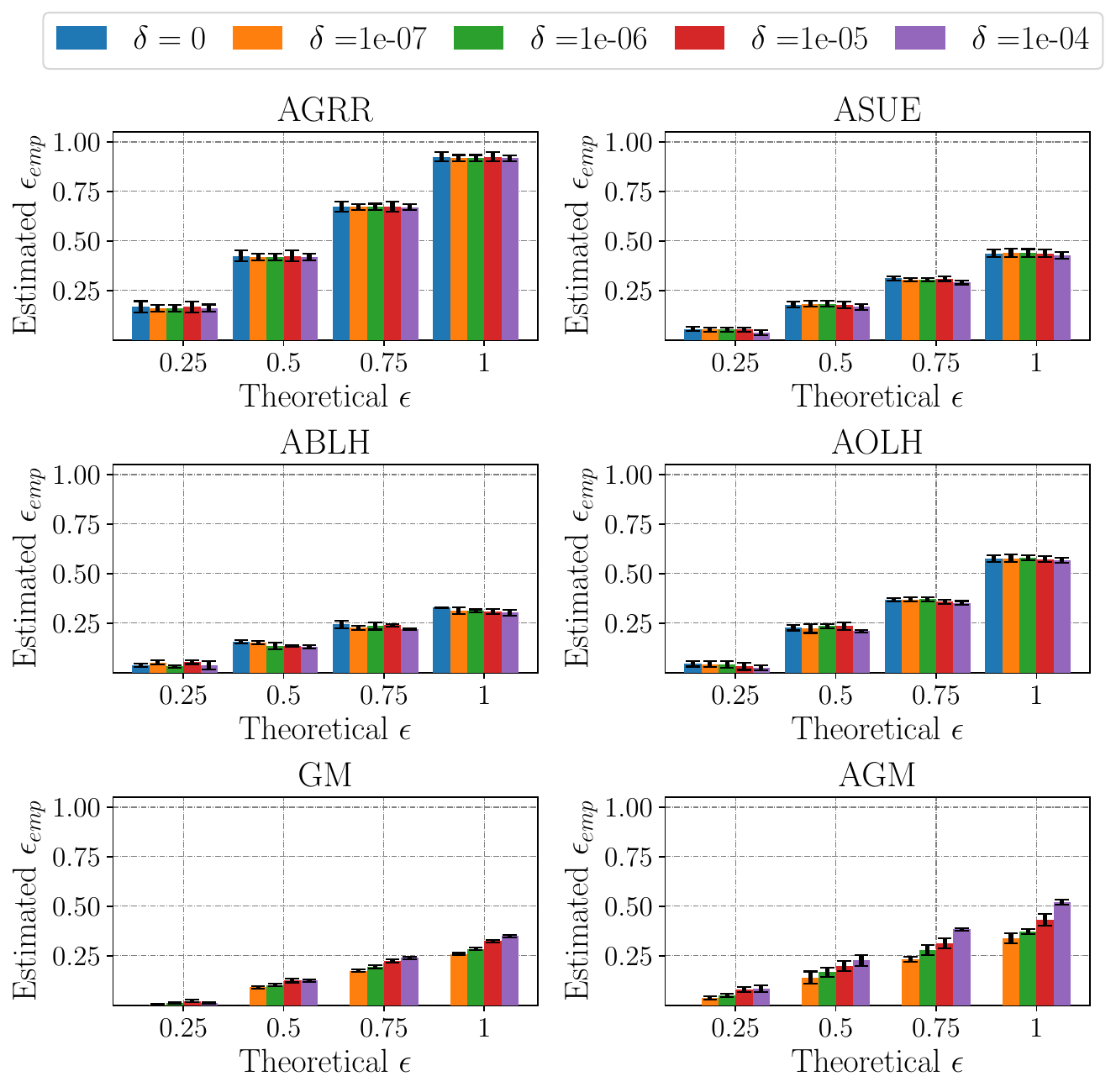}
        \caption{Domain size $k=200$.}
        \label{subfig:audit_delta_k200}
    \end{subfigure}
    \caption{Theoretical $\epsilon$ values (x-axis) versus estimated $\epslb$ values (y-axis) using our LDP-Auditor framework.
    We assess different privacy guarantees for six ($\epsdelta$)-LDP protocols across domain sizes $k \in \{25, 200\}$. 
    The special case $\delta=0$ corresponds to pure $\epsilon$-LDP, for which GM and AGM do not satisfy.}
    \label{fig:audit_delta}
\end{figure}

\subsection{Case Study \#2: Auditing the Privacy Loss of Local Hashing Encoding Without LDP} \label{sub:lh_audit}

As discussed previously in Section~\ref{sub:main_results}, both LH protocols present the least tight estimates for $\epslb$ in high privacy regimes.
Even worse, BLH's estimated privacy loss remains below $\epslb<1$ for $\epsilon\geq 2$, leading to empirical privacy losses $\leq 10$x of the theoretical $\epsilon$.
Motivated by these observations, we performed an additional study \textbf{to audit the impact of local hashing encoding but with no LDP perturbation} (\ie, $\epsilon=+\infty$), which we refer to as Local Hashing Only (LHO).
More precisely, the LHO reporting mechanism is $\mathrm{LHO}(v) \coloneqq \langle \mathrm{H}, \mathrm{H}(v) \rangle$, and we used the same distinguishability attack $\calA_{\mathrm{LH}}$ described in Section~\ref{sec:background} to attack LHO.
For these experiments, we use the following parameter values:

\begin{itemize}
    \item \textbf{LHO hash domain.} We vary the hash domain $[g]$ within the range $g \in \{2, 4, 6, 8, 10\}$.

    \item \textbf{Domain size.} We vary the domain size within the range $k \in \{25, 50, 100, 150, 200\}$. 
\end{itemize}

\begin{figure}[!htb]
    \centering
    \includegraphics[width=0.7\linewidth]{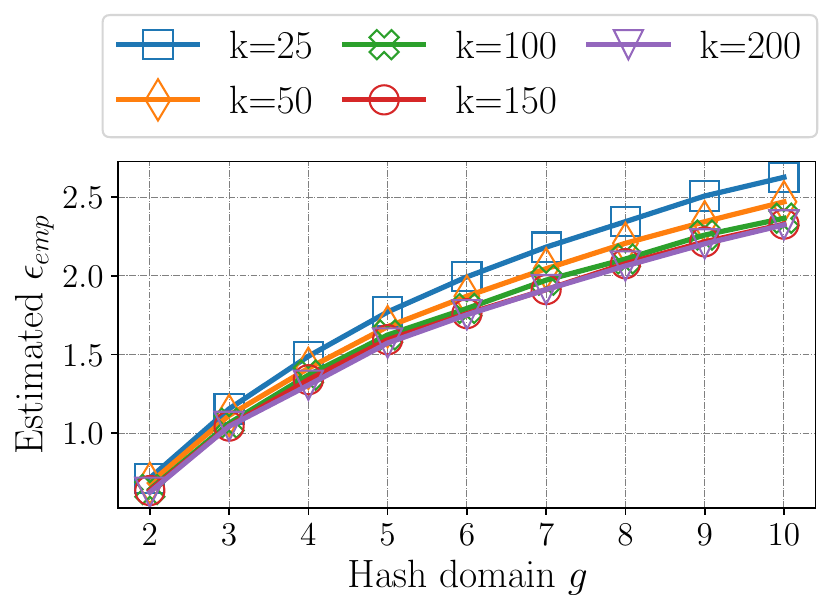}
    \caption{Estimated $\epslb$ (y-axis) versus hash domain $g$ (x-axis) using our LDP-Auditor framework comparing different domain sizes $k$ for LH encoding with no LDP randomization.}
    \label{fig:audit_lh}
\end{figure}

Figure~\ref{fig:audit_lh} presents the estimated $\epslb$ values (x-axis) for LHO protocols according to the hash domain sizes $g$ (y-axis) using our LDP-Auditor framework for different domain sizes $k$.
Observations from Figure~\ref{fig:audit_lh} underscore that, even for a binary hash domain ($g=2$), the estimated privacy loss remains $\epslb < 1$, aligning with high privacy regimes suitable for real-world applications.
Indeed, even though there is no LDP randomization of the hashed value $h \in \{0,1\}$, the adversary still has a random guess on the support set $\setOne_{\mathrm{LH}}$.
Given a general (universal) family of hash functions $\mathscr{H}$, each input value $v \in V$ is hashed into a value in $[g]$ by a hash function $\mathrm{H} \in \mathscr{H}$, and the universal property requires $\forall{v_1, v_2} \in V, v_1 \neq v_2 : \quad \underset{\mathrm{H} \in \mathscr{H}}{\Pr}\left[\mathrm{H}(v_1) = \mathrm{H}(v_2)  \right] \leq \frac{1}{g}$.
In other words, approximately $k/g$ values can be mapped to the same hashed value $h=\mathrm{H}(v)$ in $[g]$.
Although local hashing pre-processing by itself has no proven DP guarantees, this significant loss of information in the encoding step suggests potential privacy gains for LH protocols due to the presence of many random collisions.
In a similar context, DP-Sniper~\cite{bichsel2021dpsniper}, a method developed to finds violations of DP, also encountered difficulties estimating $\epsilon$ for the original RAPPOR~\cite{rappor}, which is based on Bloom filters and employs hash functions.

One could expect a similar privacy gain for other LDP mechanisms based on sketching such as Apple's Count-Mean Sketch (CMS)~\cite{apple} and Hadamard~\cite{Hadamard} mechanisms, which we leave as for future audit investigations.
Furthermore, as we increase the hash domain size $g>2$ without introducing any LDP perturbation, the estimated $\epslb$ starts to rise, achieving medium privacy regimes  $1< \epslb \leq 2.5$.
This outcome is expected since preserving more information during the encoding step decreases the support set size $|\setOne_{\mathrm{LH}}|$, which naturally enhances the accuracy of the distinguishability attack $\calA_{\mathrm{LH}}$.
Therefore, the estimated privacy loss $\epslb$ for LH-based protocols will be lower if the domain size $k$ is high and/or if the new hashed domain $g$ is small.

\subsection{Case Study \#3: Auditing the LDP Sequential Composition in Longitudinal Studies} \label{sub:audit_seq_comp_long}

As discussed in Section~\ref{sub:ldp_auditor_long}, we aim to audit the empirical privacy loss of LDP protocols in longitudinal studies (\ie, $\tau$ data collections).
This will allow to assess the gap between empirical local privacy loss estimation and the theoretical upper bound imposed by the (L)DP sequential composition.
For these experiments, we use both Algorithms~\ref{alg:ldp_auditor_lb} and~\ref{alg:attack_ldp_long} with the following parameter values:

\begin{itemize}
    \item \textbf{LDP protocols.} We audit the eight $\epsilon$-LDP protocols from Section~\ref{sub:pure_ldp_protocols}.
    Additionally, in light of the findings presented in Section~\ref{sub:audit_delta}, we only audit two ($\epsdelta$)-LDP protocols that exhibit sensitivity to $\delta$; namely, GM and AGM.
    
    \item \textbf{Number of data collections.} We vary the number of data collections in the range $\tau \in \{5, 10, 25, 50, 75, 100, 250, 500\}$.
    
    \item \textbf{Theoretical upper bound.} We vary the per-report privacy guarantee in high privacy regimes, in the range $\epsilon \in \{0.25, 0.5, 0.75, 1\}$.
    By the sequential composition, the theoretical upper bound after $\tau$ data collections is $\tau \epsilon$-LDP.

    \item \textbf{Delta parameter.} For approximate LDP, we set $\delta=1e^{-5}$.
    
    \item \textbf{Domain size.} We vary the domain size $k \in \{2, 25, 50, 100\}$. We present results for $k \in \{2, 100\}$ in the main paper and defer the others to Appendix~\ref{app:add_exp_audit_seq_comp_long}.
\end{itemize}

Figure~\ref{fig:audit_seq_comp_long} illustrates the estimated $\epslb$ values (y-axis) for the eight $\epsilon$-LDP and both GM and AGM ($\epsdelta$)-LDP protocols according to the the number of data collections $\tau$ (x-axis), per report $\epsilon$ and domain size $k \in \{2, 100\}$, using our LDP-Auditor framework. 
From Figure~\ref{subfig:audit_seq_comp_long_k2}, one can notice that both GRR and SS protocols have equal $\epslb$ estimates, as for $k=2$, the subset size $\omega=1$ (\ie, GRR).
These two LDP protocols exhibited the tightest empirical privacy estimates for $\epslb$, aligning with the observations made in Section~\ref{sub:main_results} (see Figure~\ref{fig:audit_pure_ldp}).
In contrast, the approximate LDP protocols, notably GM and AGM, showed less favorable estimates for privacy loss, which corroborates the findings illustrated in Figure~\ref{fig:audit_approx_ldp_protocols} in Section~\ref{sub:main_results}.
The remaining pure-LDP protocols -- SUE, OUE, BLH and OLH -- display intermediate privacy loss estimates.

Furthermore, Figure~\ref{subfig:audit_seq_comp_long_k100} reveals that, for a larger domain size of $k=100$, the results obtained are reversed.
Among pure-LDP protocols, GRR yields the lowest $\epslb$ estimation for all experimented $\tau$ values, followed by the SHE protocol.
The reason for this is that the probability of being ``honest'' $p=\frac{e^{\epsilon}}{e^{\epsilon} + k - 1}$ in Equation~\eqref{eq:grr}, is directly proportional to the domain size $k$.
Therefore, even after many data collections $\tau$, the adversary has still too much noisy data to filter, which makes the distinguishability attack less efficient.
Similar to Figure~\ref{fig:audit_seq_comp_long}, in Figure~\ref{fig:add_exp_audit_seq_comp_long}, approximate-LDP protocols (GM and AGM) led to the lowest empirical privacy loss estimates for $\epslb$.

\begin{figure*}[!htb]
    \centering
    \begin{subfigure}{\columnwidth}
        \includegraphics[width=\linewidth]{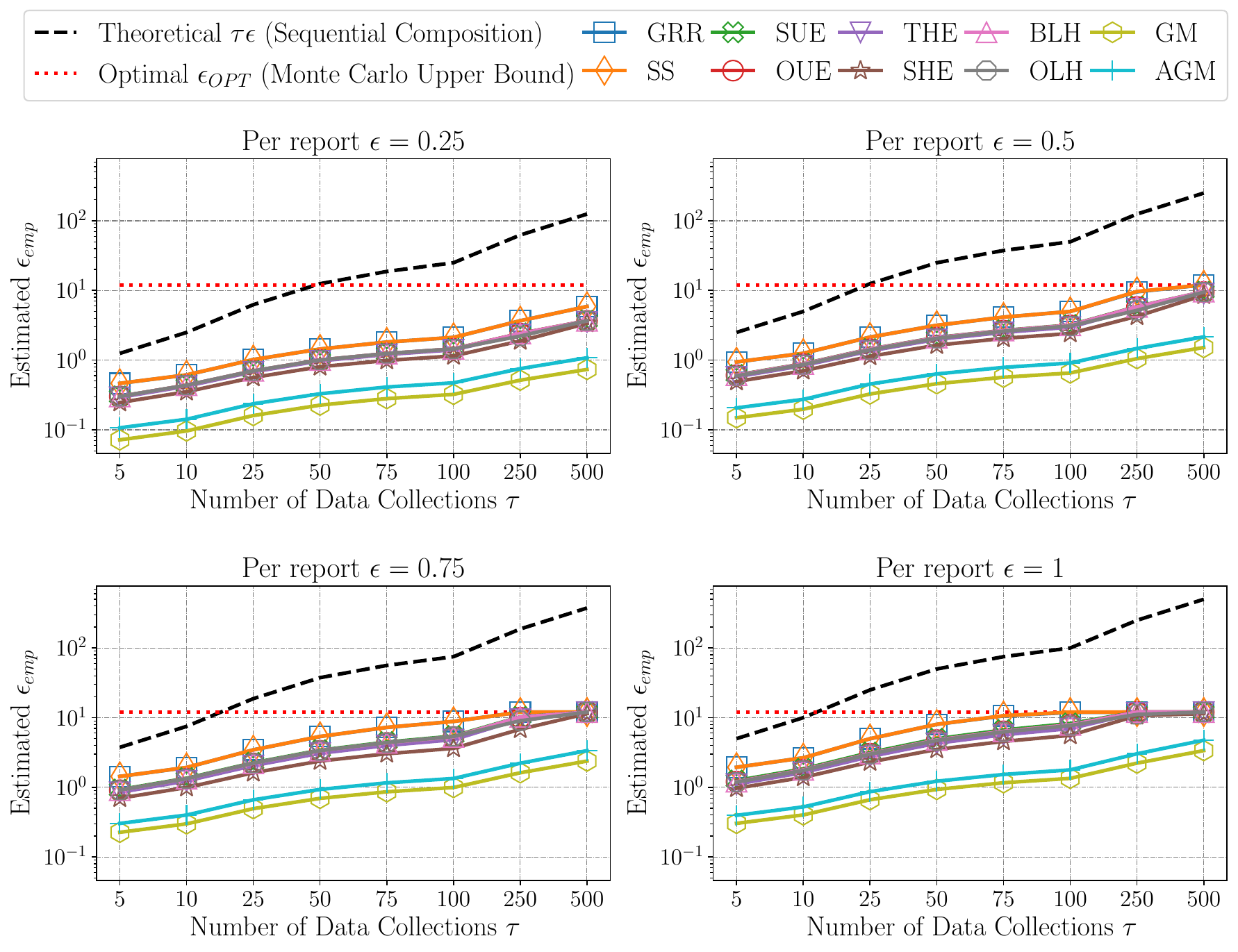}
        \caption{Domain size $k=2$.}
        \label{subfig:audit_seq_comp_long_k2}
    \end{subfigure}
    \hfill 
    \begin{subfigure}{\columnwidth}
        \includegraphics[width=\linewidth]{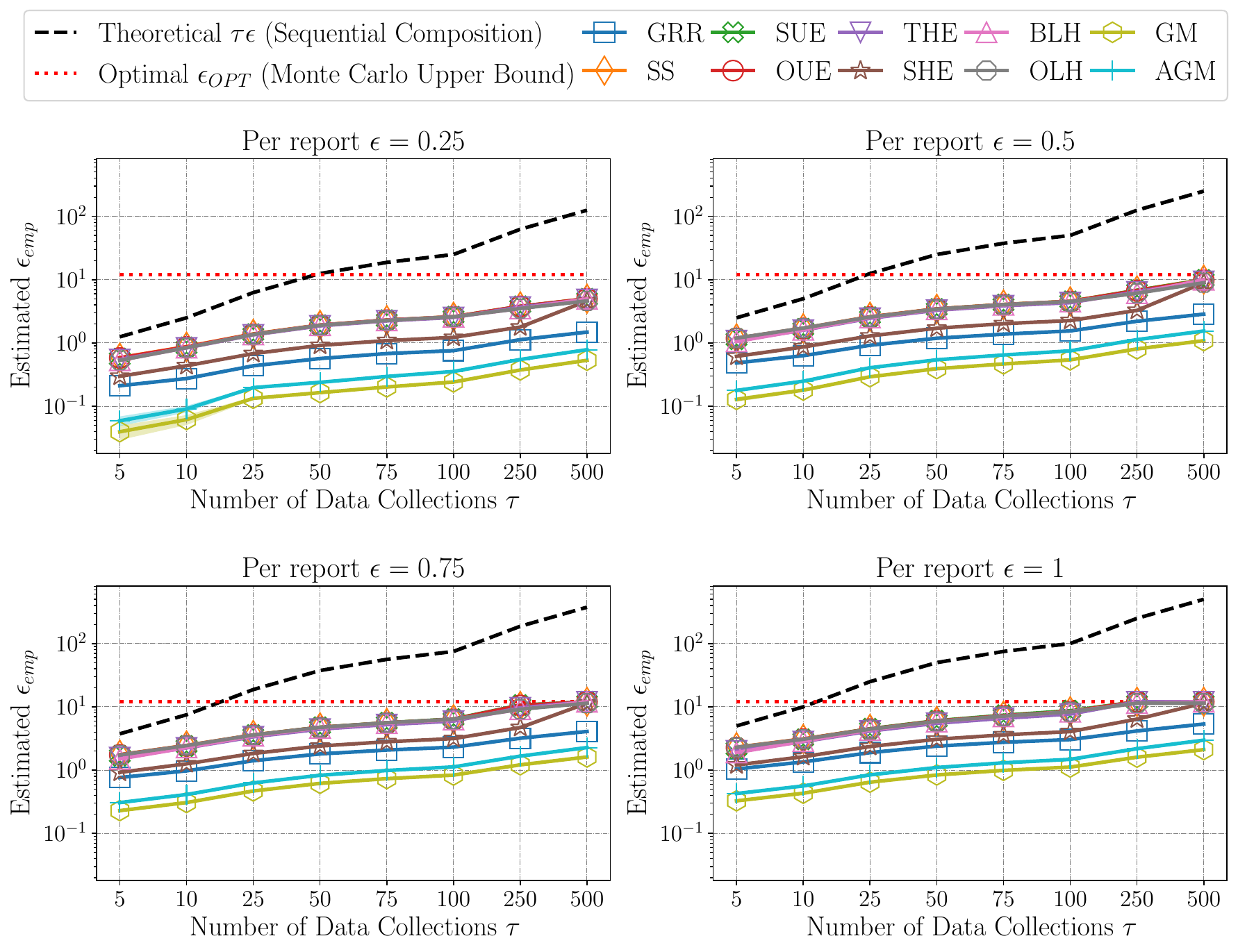}
        \caption{Domain size $k=100$.}
        \label{subfig:audit_seq_comp_long_k100}
    \end{subfigure}
    \caption{Estimated $\epslb$ (y-axis) versus the number of data collections $\tau$ (x-axis) using our LDP-Auditor framework for different domain sizes $k \in \{2, 100\}$. 
    We vary the per report $\epsilon$-LDP guarantee for the following LDP frequency estimation protocols: GRR, SS, SUE, OUE, BLH, OLH, SHE, THE, GM and AGM.
    For both approximate-LDP protocols, namely GM and AGM, $\delta=1e^{-5}$.}
    \label{fig:audit_seq_comp_long}
\end{figure*}

Moreover, from both Figure~\ref{fig:audit_seq_comp_long} and Figure~\ref{fig:add_exp_audit_seq_comp_long}, it is evident that even after $\tau=500$, none of the LDP protocols, achieves the optimal upper bound $\epsopt$ imposed by the Monte Carlo estimation when the per-report privacy guarantee is too small (\ie, $\epsilon=0.25$). 
However, as the number of data collections becomes sufficiently large (\ie, $\tau \geq 250$) and the privacy guarantee per report also increases (\eg, $\epsilon \geq 0.75$), all pure-LDP protocols, with the exception of GRR, manage to achieve the Monte Carlo upper bound, resulting in $\epslb=\epsopt$.
Yet, as the number of data collections becomes sufficiently large (\ie, $\tau \to \infty$), we anticipate that $\epslb$ will converge to $\epsopt$ for all LDP protocols even when the per-report $\epsilon<0.5$.

These results are quite surprising since one would imagine the privacy leakage to be higher for repeated data collections when random fresh noise is added per report.
Nevertheless, as the domain size increases, the performance of the distinguishability attack decreases~\cite{Gursoy2022,Arcolezi2023}.
As a consequence, for real-world deployments with substantial domain sizes (\eg, list of Internet domains), exclusively relying on theoretical $\epsilon$-LDP guarantees may prove unrealistic. 
Privacy auditing becomes imperative in such scenarios, to establish appropriate privacy parameters, thus avoiding adding more noise than required. 
\textit{Notably, these auditing results emphasize a crucial aspect for longitudinal studies: a substantial gap exists between theory (sequential composition) and practice (LDP auditing).}
To narrow this gap, one could consider designing more powerful attacks for longitudinal studies beyond those proposed here in Algorithm~\ref{alg:attack_ldp_long}. 
Alternatively, research efforts could be directed towards developing more sophisticated compositions for $\epsilon$-LDP mechanisms.

\subsection{Case Study \#4: LDP Auditing with Multidimensional Data} \label{sub:audit_multidimensional}

As discussed in Section~\ref{sub:ldp_auditor_multidimensional}, our audit results outlined in Section~\ref{sub:main_results} are also valid for LDP mechanisms based on the standard SPL and SMP solutions for multidimensional data.
Thus, in this section, we aim to audit LDP protocols following the RS+FD~\cite{Arcolezi_rs_fd} solution.
For these experiments, we use both Algorithms~\ref{alg:attack_rs+fd} and~\ref{alg:ldp_auditor_rs+fd}, considering:

\begin{itemize}
    \item \textbf{LDP protocols.} We audit five $\epsilon$-LDP RS+FD protocols: RS+FD[GRR], RS+FD[SUE-z], RS+FD[SUE-r], RS+FD[OUE-z] and RS+FD[OUE-r].
    The difference between UE-z and UE-r lies on how to generate the fake data~\cite{Arcolezi_rs_fd}.
    More precisely, UE-z initializes a zero-vector and UE-r initializes a random one-hot-encoded vector.
    Next, SUE or OUE is used to sanitize these vectors.

    \item \textbf{Theoretical upper bound.} We vary the theoretical privacy parameter $\epsilon$ in high, mid and low privacy regimes over the same range $\epsilon \in \{0.25, 0.5, 0.75, 1, 2, 4, 6, 10\}$ as in Section~\ref{sub:main_results}.
        
    \item \textbf{Domain size and number of attributes.} We vary the domain size as $k \in \{2, 25, 50, 100\}$ and we vary the number of attributes over $d\in \{2, 10\}$.
    When $d=2$, $\mathbf{k}=[2, 2]$, $\mathbf{k}=[25, 25]$, $\mathbf{k}=[50, 50]$ and $\mathbf{k}=[100, 100]$ and, in a similar way for $d=10$.
    We present results for $k \in \{2, 100\}$ in the main paper and defer the others to Appendix~\ref{app:add_exp_audit_multidimensional}.
\end{itemize}

\begin{figure*}[!h]
    \centering
    \begin{subfigure}{1\columnwidth}
        \includegraphics[width=\linewidth]{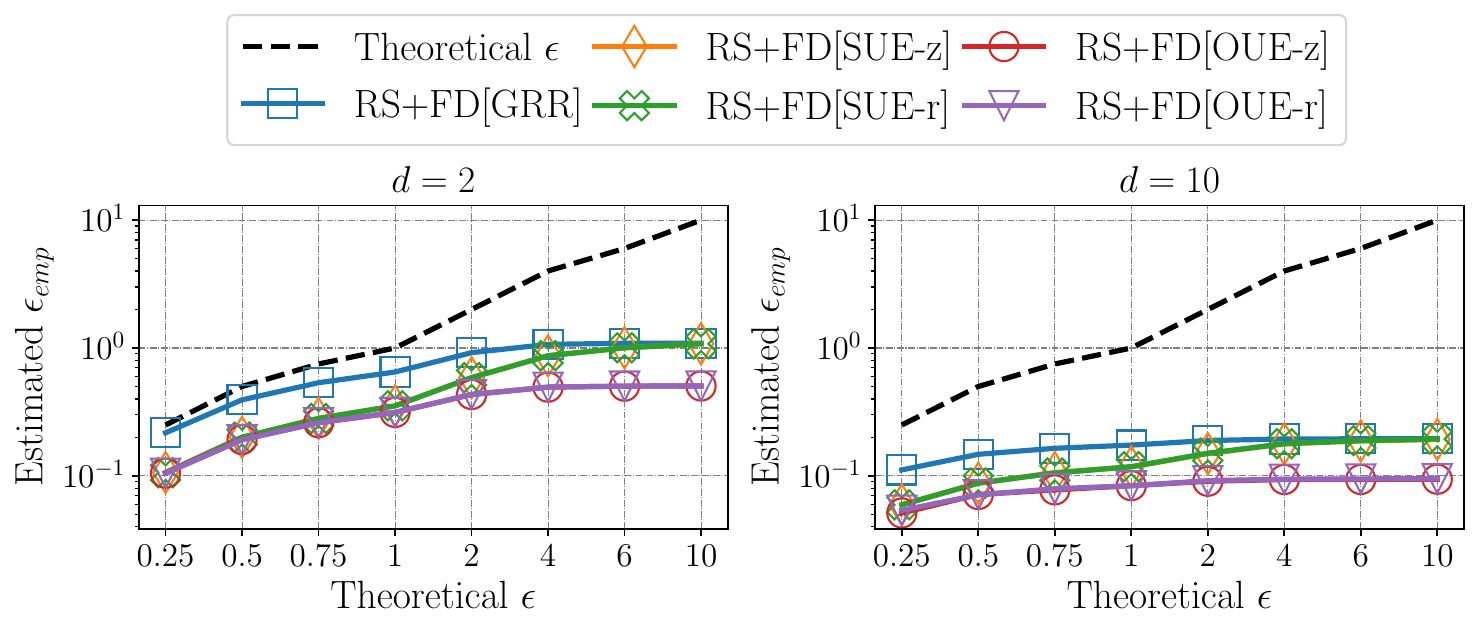}
        \caption{Domain size $k=2$.}
        \label{fig:audit_rs+fd_k2}
    \end{subfigure}
    \hfill 
    \begin{subfigure}{1\columnwidth}
        \includegraphics[width=\linewidth]{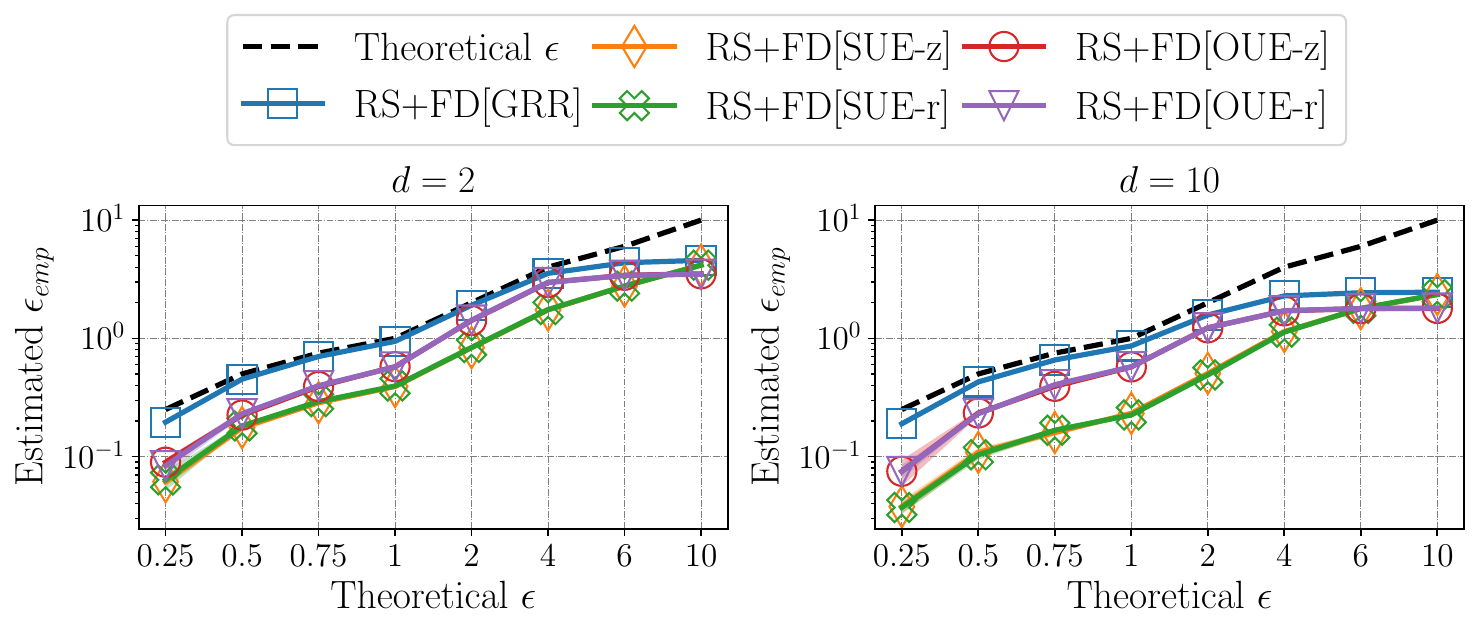}
        \caption{Domain size $k=100$.}
        \label{fig:audit_rs+fd_k100}
    \end{subfigure}
    \caption{Theoretical $\epsilon$ (x-axis) versus estimated $\epslb$ (y-axis) using our LDP-Auditor framework comparing different number of attributes $d$ for five RS+FD~\cite{Arcolezi_rs_fd} protocols with domain sizes $k=2$ and $k=100$.}
    \label{fig:audit_rs+fd}
\end{figure*}

Figure~\ref{fig:audit_rs+fd} illustrates the comparison of theoretical $\epsilon$ values (x-axis) with estimated $\epslb$ values (y-axis) for the five RS+FD protocols, based on the number of attributes $d$ and domain size $k \in \{2, 100\}$, utilizing our LDP-Auditor framework.
From Figure~\ref{fig:audit_rs+fd}, it is clear that, once again, GRR exhibits tighter empirical privacy losses $\epslb$ than UE-based protocols following the RS+FD solution. 
However, in contrast to Section~\ref{sub:main_results}, the estimated $\epslb$ for GRR now displays a ``plateau behaviour'' after theoretical $\epsilon\geq 4$. 
This plateau arises because the probability of reporting the true value under GRR reaches high values with $\epsilon\geq 4$. 
Notably, among the family of UE protocols, SUE demonstrates a tighter empirical $\epslb$ than OUE when the domain is binary (see Figure~\ref{fig:audit_rs+fd_k2}). 
However, SUE exhibits lower $\epslb$ than OUE when $k=100$ (see Figure~\ref{fig:audit_rs+fd_k100}). 
This observation can be attributed to the advantage of SUE in transmitting the true bit with a probability $p>\frac{1}{2}$, while OUE has $p=\frac{1}{2}$. 
Consequently, the distinguishability attack achieves higher accuracy for SUE, increasing the true positive rate and decreasing the false positive rate, resulting in higher $\epslb$ estimates.
Moreover, different fake data generation procedures for UE protocols (UE-z \emph{vs} UE-r) did not result in significant changes in the audit results.

Another intriguing result is that the empirical privacy loss is lower for a binary domain compared to when $k=100$. 
This behavior is primarily due to the impact of fake data on distinguishability attacks. 
\textit{In a binary domain, fake data significantly increases the false positive rate, leading to a decrease in the estimated privacy loss $\epslb$.}
However, for a higher domain size, fake data has a lesser impact on the false positive rate, as the distinguishability attack has more rooms for errors.
Overall, these nuanced relationships underscore the intricate interplay between domain size, the use of fake data and the tightness of local privacy loss estimation in the context of RS+FD protocols.

\subsection{Case Study \#5: Debugging a Python Implementation of UE Protocols} \label{sub:debugging}

Finally, we show how our LDP-Auditor framework can also serve as a tool for verifying the correctness of LDP implementations. 
In our case study, we focused on the pure-LDP~\cite{pureldp} package (version 1.1.2) and show that their UE protocols fail to meet the claimed level of $\epsilon$-LDP.
Our objective here is not to point out issues with respect to a particular code or library but rather to demonstrate the potentiality of our approach for verifying and debugging LDP protocols.
Following a similar experimental setup as the one outlined in Section~\ref{sub:main_results}, Figure~\ref{fig:audit_pure_ldp} presents a comparison of the theoretical $\epsilon$ values (x-axis) with the estimated $\epslb$ values (y-axis) using our LDP-Auditor framework. 
We consider different domain sizes $k$ for both the SUE and OUE protocols, implemented in the pure-LDP package. 
The inconsistencies we found between the lower and upper bounds are highlighted within the \textcolor{orange}{orange rectangle}.

\begin{figure}[h!]
    \centering
    \includegraphics[width=1\linewidth]{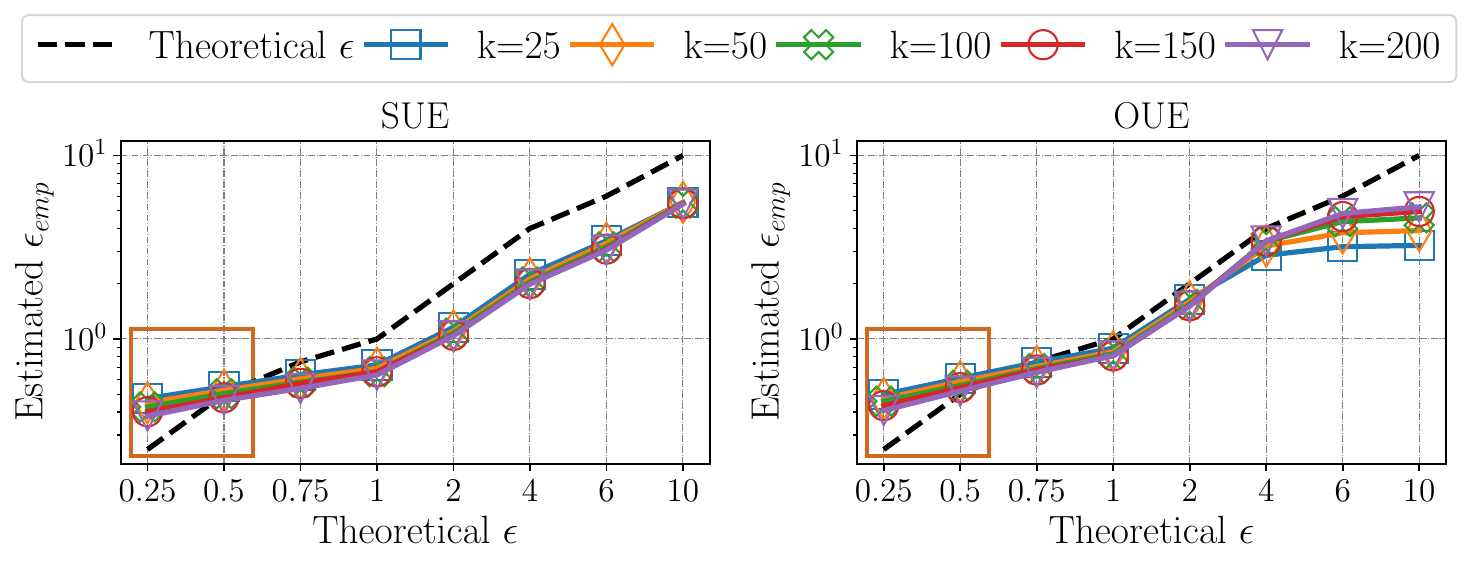}
    \caption{Theoretical $\epsilon$ (x-axis) versus estimated $\epslb$ (y-axis) using our LDP-Auditor framework comparing different domain sizes $k$ for both SUE and OUE protocols, implemented in the pure-LDP package~\cite{pureldp}. 
    The orange rectangle highlights inconsistencies between the observed empirical privacy loss and the theoretical upper bound.}
    \label{fig:audit_pure_ldp}
\end{figure}

From Figure~\ref{fig:audit_pure_ldp}, it is clear that LDP-Auditor has detected inconsistencies between the lower and upper bounds, which are highlighted by the orange rectangle. 
After conducting an investigation into the pure-LDP code, we were able to identify the specific location of the implementation error. 
The error arises from the following steps in the \texttt{\_perturb} function of the \texttt{UEClient} class:

\begin{enumerate}

    \item The user initializes a zero-vector $\textbf{y}=[0, 0, \ldots, 0]$ of size $k$;
    
    \item The user samples indexes of values in $\textbf{y}$ that will flip from $0$ to $1$ with probability $q$ (as indicated in Equation~\eqref{eq:ue_parameters}).

    \item With probability $p$ (as indicated in Equation~\eqref{eq:ue_parameters}), the index at position $\textbf{y}_v$ (representing the user's true value) is flipped from $0$ to $1$.

    \item \textcolor{teal}{\textbf{*Missing step*:}} if $\textbf{y}_v$ was set to $1$ in step (2) but not in step (3), there should be a correction to revert it back to $0$.
\end{enumerate}

This was a simple mistake that \textbf{\textit{was directly fixed by the authors~\cite{fix_bug_pure_ldp} following our communication with them.}}
However, it is crucial to emphasize that this minor error had implications for the $\epsilon$-LDP guarantees. 
Specifically, the bit corresponding to the user's value was transmitted more time than intended, particularly in high privacy regimes.
In mid to low privacy regimes, the bug might go unnoticed, given the already high probability of transmitting the bit as 1.
This explains why LDP-Auditor failed to detect inconsistencies between the empirical and upper bounds for $\epsilon \geq 1$. 
In such cases, specialized tools designed for identifying DP violations, like DP-Sniper~\cite{bichsel2021dpsniper}, would likely have been effective in detecting the bug.
Therefore, we strongly encourage end-users of the pure-LDP package to update to the latest version 1.2.0.

\section{Conclusion and Perspectives} 
\label{sec:conclusion}

In this work, we have introduced the LDP-Auditor framework as a powerful tool for empirically estimating the privacy loss of LDP frequency estimation protocols. 
Our main LDP audit results provide new insights into the empirical local privacy loss in practical adversarial settings.
Through several case studies, we have demonstrated the framework's effectiveness in identifying significant discrepancies between theoretical guarantees and empirical privacy loss. 
These findings contribute to a nuanced understanding of the challenges and considerations in the design and implementation of LDP mechanisms. 
As LDP continues to gain prominence in privacy-preserving data analysis, LDP-Auditor can serve as a valuable resource for practitioners and researchers aiming to assess and enhance the privacy guarantees of their systems.

Nevertheless, while we instantiated LDP-Auditor with distinguishability attacks on the user's value~\cite{Gursoy2022, Arcolezi2023}, our future plans involve expanding the scope of LDP auditing to incorporate other adversarial analysis proposed in the literature, such as inference pool~\cite{Gadotti2022}, data change detection~\cite{Arcolezi2023evolving} and re-identification attacks~\cite{Murakami2021, Arcolezi2023}. 
We also aim and suggest extending LDP-Auditor to encompass a wider range of LDP applications (\eg, mean estimation) as these may introduce unique challenges and considerations during the auditing process.
Additionally, we aim at integrating the Neyman-Pearson lemma into LDP-Auditor's analysis to leverage its theoretical foundation to enhance the precision of our auditing framework.
Lastly, one can envision utilizing LDP-Auditor as a means to establish a unified local privacy loss $\epslb$ when comparing mechanisms of different locally private definitions, such as $d$-privacy~\cite{chatzikokolakis2013broadening}, $\alpha$-PIE~\cite{Murakami2021} and LDP~\cite{first_ldp}.

\begin{acks}
We thank Catuscia Palamidessi, Aurélien Bellet and Mathias Lécuyer for their helpful discussions and feedback throughout this project. 
The authors also deeply thank the anonymous PETS reviewers for their insightful suggestions.
This work has been partially supported by the ``ANR 22-PECY-0002'' IPOP (Interdisciplinary Project on Privacy) project of the Cybersecurity PEPR.
The work of H\'eber H. Arcolezi was partially supported by the European Research Council (ERC) project HYPATIA under the European Union’s Horizon 2020 research and innovation programme. Grant agreement n. 835294.
Sébastien Gambs is supported by the Canada Research Chair program as well as a Discovery Grant from NSERC. 
\end{acks}

\bibliographystyle{ACM-Reference-Format}
\bibliography{ref.bib}


\appendix

\clearpage
\section{Summary of Notations} \label{app:notation}
The main notation used in this paper is summarized in Table~\ref{tab:notation}.

\begin{table}[!h]
    \centering
    \begin{tabular}{c l}
    \toprule
     Symbol                 & Description \\
     \toprule
     $[a]$                  & Set of integers $\{1, 2, 3, \ldots, a\}$. \\
     $\mathbf{a}_i$         & $i$-th coordinate of vector $\mathbf{a}$. \\
     $V$                    & Data domain. \\
     $k$                    & Domain size $k=|V|$. \\
     $n$                    & Number of users. \\
     $\epsilon$             & Theoretical privacy loss. \\
     $\delta$               & Maximum probability that privacy loss exceeds $\epsilon$. \\
     $\epslb$               & Empirical privacy loss. \\
     $\epsopt$              & Upper bound on Monte Carlo privacy loss. \\
     $\calM$                & ($\epsdelta$)-LDP mechanism. \\
     $\calA$                & Distinguishability attack. \\
     $\calA^{L}$            & Distinguishability attack in longitudinal study. \\
     $\calA^{\textrm{RS+FD}}$ & Distinguishability attack on RS+FD protocols. \\
     $\calA_{\calM}$        & Distinguishability attack of mechanism $\calM$. \\
     $\setOne_{\calM}$      & Support set of mechanism $\calM$. \\     
     $T$                    & Number of trials. \\
     $\alpha$               & Confidence level. \\
     $d$                    & Number of attributes $d\geq 2$.\\
     $\tau$                 & Number of data collections.\\
     \bottomrule
    \end{tabular}
    \caption{Symbols and Notations.}
    \label{tab:notation}
\end{table}

\section{Approximate ($\epsdelta$)-LDP Protocols} \label{app:approximate_ldp_protocols_detailed}

\textbf{Approximate GRR (AGRR)~\cite{Wang2021_approx_ldp}.} Similar to GRR in Section~\ref{sub:pure_ldp_protocols}, given a value $v \in V$, $\mathrm{AGRR}(v)$ outputs the true value $v$ with probability $p$, and any other value $v' \in V \setminus \{v\}$, otherwise. 
More formally:
\begin{equation} \label{eq:agrr}
    \Pr[\mathrm{AGRR}(v)=y] = \begin{cases} p=\frac{e^{\epsilon} + (k - 1)\delta}{e^{\epsilon}+k-1} \textrm{ if } y = v,\\ q=\frac{1 - \delta}{e^{\epsilon}+k-1} \quad \textrm{ if } y \neq v \textrm{,} \end{cases}
\end{equation}

\noindent in which $y \in V$ is the perturbed value sent to the server. 
From Equation~\eqref{eq:agrr}, $\Pr[y=v] > \Pr[y=v']$ for all $v' \in V \setminus \{v\}$. 
Thus, the attack strategy $\calA_{\mathrm{AGRR}}$ is equivalent to $\calA_{\mathrm{GRR}}$, \ie, to predict $\hat{v}=y$.

\textbf{Approximate SUE (ASUE)~\cite{Wang2021_approx_ldp}.}
Similar to the SUE protocol~\cite{rappor} in Section~\ref{sub:pure_ldp_protocols}, ASUE encode the user's input data $v \in V$, as a one-hot $k$-dimensional vector.
The obfuscation function of ASUE randomizes the bits from $\textbf{v}$ independently to generate $\textbf{y}$ as follows:
\begin{equation}  \label{eq:asue_parameters}
    \forall{i \in [k]} : \quad \Pr[\textbf{y}_i=1] =\begin{cases} p = \frac{e^\epsilon - \sqrt{e^\epsilon(1 - \delta) + \delta}}{e^\epsilon - 1}, \textrm{ if } \textbf{v}_i=1 \textrm{,} \\ q = \frac{\sqrt{e^\epsilon(1 - \delta) + \delta} - 1}{e^\epsilon - 1}, \textrm{      if } \textbf{v}_i=0 \textrm{,}\end{cases}
\end{equation}

\noindent in which $\textbf{y}$ is sent to the server.
As for UE protocols, with $\textbf{y}$, the adversary can construct the subset of all values $v \in V$ that are set to 1, \ie, $\setOne_{\mathrm{AUE}}=\{v | \textbf{y}_v = 1\}$.
Then, the attack strategy $\calA_{\mathrm{ASUE}}$ is equivalent to $\calA_{\mathrm{AUE}}$:
\begin{itemize}
    \item $\calA^0_{\mathrm{ASUE}}$ is a random choice $\hat{v}=\mathrm{Uniform}\left( [k] \right)$, if $\setOne_{\mathrm{AUE}}=\emptyset$;
    
    \item $\calA^1_{\mathrm{ASUE}}$ is a random choice $\hat{v}=\mathrm{Uniform}\left( \setOne_{\mathrm{AUE}} \right)$, otherwise.
\end{itemize}

\textbf{Approximate LH (ALH)~\cite{Wang2021_approx_ldp}.} Similar to the LH protocols~\cite{tianhao2017,Bassily2015} in Section~\ref{sub:pure_ldp_protocols}, ALH uses a hash function $\mathrm{H} \in \mathscr{H}$ to map the input data $v \in V$ to a new domain of size $g \geq 2$, and then apply AGRR to the hashed value $h=\mathrm{H}(v)$.
In particular, the ALH reporting mechanism is $\mathrm{ALH}(v) \coloneqq \langle \mathrm{H}, \mathrm{AGRR}(h) \rangle \textrm{,}$ in which $\mathrm{AGRR}$ is given in Equation~\eqref{eq:agrr} while operating on the new domain $[g]$. 
The two variants of ALH protocols are: (1) Approximate BLH (ABLH), which sets $g=2$ and (2) Approximate OLH (AOLH), which sets $g = \frac{-3e^\epsilon \delta - \sqrt{e^\epsilon - 1}\sqrt{(1 - \delta)(e^\epsilon + \delta - 9e^\epsilon \delta - 1)} + e^\epsilon + 3\delta - 1}{2\delta}$.
Each user reports the hash function and obfuscated value $\langle \mathrm{H}, y \rangle$ to the server. 
With these elements, the adversary can construct the subset of all values $v \in V$ that hash to $y$, \ie, $\setOne_{\mathrm{ALH}}= \{v | \mathrm{H}(v) = y\}$.
Then, the attack strategy $\calA_{\mathrm{ALH}}$ is equivalent to $\calA_{\mathrm{LH}}$:
\begin{itemize}
    \item $\calA^0_{\mathrm{ALH}}$ is a random choice $\hat{v}=\mathrm{Uniform}\left( [k] \right)$, if $\setOne_{\mathrm{ALH}}=\emptyset$;
    
    \item $\calA^1_{\mathrm{ALH}}$ is a random choice $\hat{v}=\mathrm{Uniform}\left( \setOne_{\mathrm{ALH}} \right)$, otherwise.
\end{itemize}

\section{Adversarial Privacy Game} \label{app:adv_priv_game}
Figure~\ref{fig:adversarial_privacy_game} provides a comparative illustration of the adversarial privacy game in central and local differential privacy frameworks, highlighting scenarios of membership inference and value distinguishability attacks, respectively.

\begin{figure}[!htb]
    \centering
    \begin{subfigure}{\columnwidth}
        \centering 
        \includegraphics[width=1\linewidth]{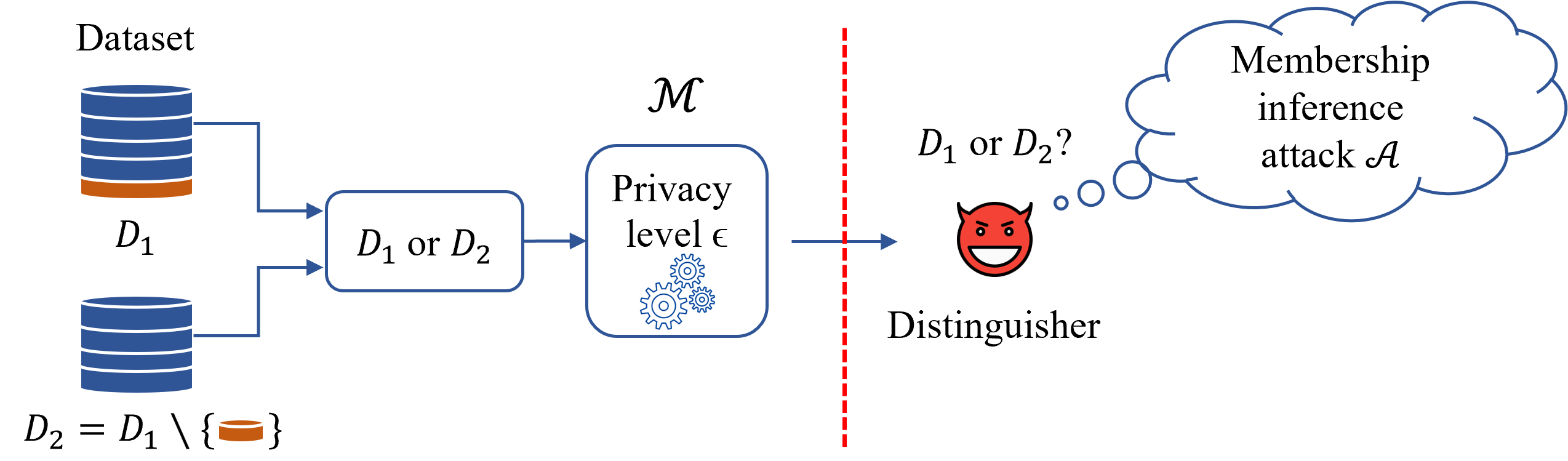}
        \caption{Central DP with membership inference attack.}
        \label{subfig:adversary_central_dp}
    \end{subfigure}
    \vspace{1em}     
    \begin{subfigure}{\columnwidth}
        \centering 
        \includegraphics[width=1\linewidth]{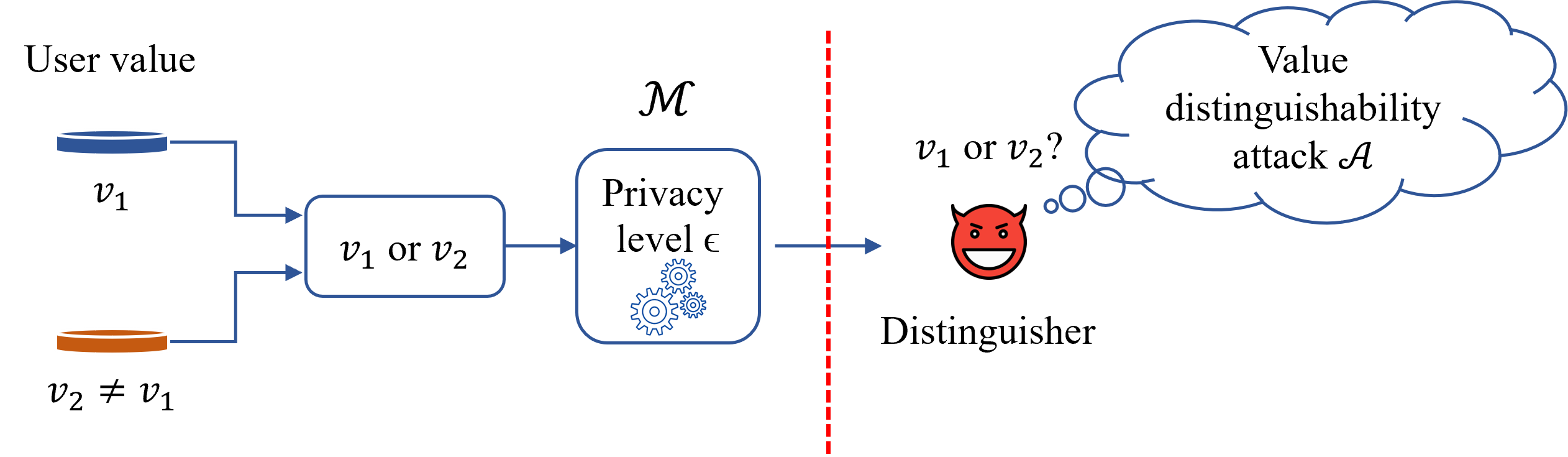}
        \caption{Local DP with value distinguishability attack.}
        \label{subfig:adversary_local_dp}
    \end{subfigure}
    
    \caption{Comparison of the adversarial privacy game between the central and local DP settings.}
    \label{fig:adversarial_privacy_game}
\end{figure}

\section{Clopper-Pearson Interval} \label{app:clopper_pearson}

The Clopper-Pearson method~\cite{clopper1934use} is a statistical technique used to calculate exact confidence intervals for the success probability in binomial distributions. 
This method is known for its conservative nature, ensuring that the confidence interval computed does not rely on any asymptotic approximations and is therefore valid regardless of the sample size. 
Given $x$ successes in $T$ trials, the Clopper-Pearson interval computes the lower and upper confidence limits for the true probability of success, based on the beta distribution's cumulative density function.
Specifically, the Clopper-Pearson confidence interval is computed as follows:

\begin{equation}
     \left[ \mathfrak{B}\left(\frac{\alpha}{2}; x, T - x + 1\right), \mathfrak{B}\left(1 - \frac{\alpha}{2}; x + 1, T - x\right) \right]   \textrm{,}
\end{equation}

\noindent in which $\mathfrak{B}$ denotes the beta distribution quantile function, $x$ is the number of observed successes, $T$ is the total number of trials, $\alpha$ represents the significance level and $\mathfrak{B} (p ; z, w)$ is the $p$-th quantile from a beta distribution with shape parameters $z$ and $w$. 
This exact method is crucial in our LDP auditing framework, as it allows us to establish the lower and upper bounds for the true positive rate and false positive rate of Equation~\eqref{eq:ldp_audit} with high confidence, ensuring that our empirical privacy loss estimations are both accurate and robust.

\section{Proof of Theorem~\ref{theorem:ldp_auditor}} \label{app:proof_theorem_1}
\begin{proof}[Proof of Theorem~\ref{theorem:ldp_auditor}]
First, the guarantee of the Clopper-Pearson confidence intervals is that, with probability at least $1 - \alpha$, $\hat{p}_0 \leq p_0$ and $\hat{p}_1 \geq p_1$, which implies $p_0/p_1 \geq \hat{p}_0/\hat{p}_1$. 
Second, if $\calM$ is ($\epsdelta$)-LDP, then we would have $p_0 \leq p_1 e^{\epsilon} + \delta$, meaning $\calM$ is not $(\epsilon',\delta)$-LDP for any $\epsilon' < \ln((p_0 - \delta)/p_1)$. 
Combining the two statements, $\calM$ is not $\epsilon'$ for any $\epsilon' < \ln((\hat{p}_0 - \delta)/\hat{p}_1) = \epslb$.
\end{proof}

\section{Memoization-Based LDP Protocols} \label{app:memoization_ldp}

As mentioned in Section~\ref{sub:ldp_auditor_long}, in longitudinal studies, the privacy loss is linear on the number of data collections $\tau$ following the DP sequential composition.
This accumulation allows attackers to employ ``averaging attacks'' to more easily distinguish a user's true value among the noisy data. 
To counteract this, renowned LDP mechanisms for longitudinal studies, such as RAPPOR~\cite{rappor} and $d$BitFlipPM~\cite{microsoft}, incorporate a \textit{memoization-based} strategy.

One way to employ memoization is to memorize an obfuscated value $y=\calM(v)$ and consistently reuse it throughout time~\cite{microsoft,Arcolezi2021}.
Specifically, at each time $t \in [\tau]$, the user reports the memorized $y$, which satisfies $\epsilon$-LDP. 
Note that as there is only a single obfuscation round, our LDP-Auditor operates equivalently to auditing in a single data collection scenario (\ie, Algorithm~\ref{alg:ldp_auditor_lb}).

An alternative memoization technique involves re-using the memorized obfuscated value $y=\calM(v)$ as the input for a subsequent round of obfuscation~\cite{rappor,Vidal2020,Arcolezi2022,Arcolezi2023}.
This means that at each time $t \in [\tau]$ the user reports $y^t=\calM(y)$; note that the input to $\calM$ is an already obfuscated value $y$.
In this setting, there are two levels of privacy guarantees~\cite{rappor}: $\epsilon_{1}$, which is the privacy level of the first report $y^1=\calM(y)$ following the second obfuscation round, and $\epsilon_{\infty}$, which is the privacy guarantee offered by the first obfuscation round that generated $y$.
More precisely, $y=\calM(v)$ satisfy $\epsilon_{\infty}$-LDP because it establishes the upper bound for the privacy leakage as an adversary could only recover $y$ instead of $v$ after executing an ``averaging attack'' across an indefinite number of reports $y^1, y^2, \ldots, y^{\infty}$.
Consequently, our LDP-Auditor framework described in Algorithm~\ref{alg:ldp_auditor_lb} can be deployed directly to estimate an empirical privacy loss $\epslb$ against the theoretical upper bound $\epsilon_{1}$ for a single data collection.
For $t \to \infty$ data collections, the theoretical upper bound becomes $\epsilon_{\infty}$, for which the distinguishability attack in longitudinal study $\calA^{L}$ outlined in Algorithm~\ref{alg:attack_ldp_long} should be applied.
In other words, while in Section~\ref{sub:audit_seq_comp_long} the upper bound is $\tau \epsilon$-LDP, for memoization-based mechanisms with two obfuscation rounds, the upper bound is $\epsilon_{\infty}$-LDP.

\section{Additional Experiments} \label{app:add_exp}

\subsection{Case Study \#1: Auditing the Impact of $\delta$} \label{app:add_exp_impact_delta}

Following the experimental setup detailed in Section~\ref{sub:audit_delta}, Figure~\ref{fig:appendix_audit_delta} illustrates the theoretical $\epsilon$ values (x-axis) versus estimated $\epslb$ values (y-axis) when varying the $\delta$ parameter and domain size $k \in \{100, 150\}$, using our LDP-Auditor framework. 
Note that for both GM and AGM protocols, there is no $\epslb$ value when $\delta = 0$, as these protocols do not have pure $\epsilon$-LDP variations.
Finally, a similar trend as in Figure~\ref{fig:audit_delta} can be observed in Figure~\ref{fig:appendix_audit_delta}, for which the discussion in Section~\ref{sub:audit_delta} is equally applicable to these results.

\begin{figure*}[h!]
    \centering
    \begin{subfigure}{\columnwidth}
        \includegraphics[width=1.0\linewidth]{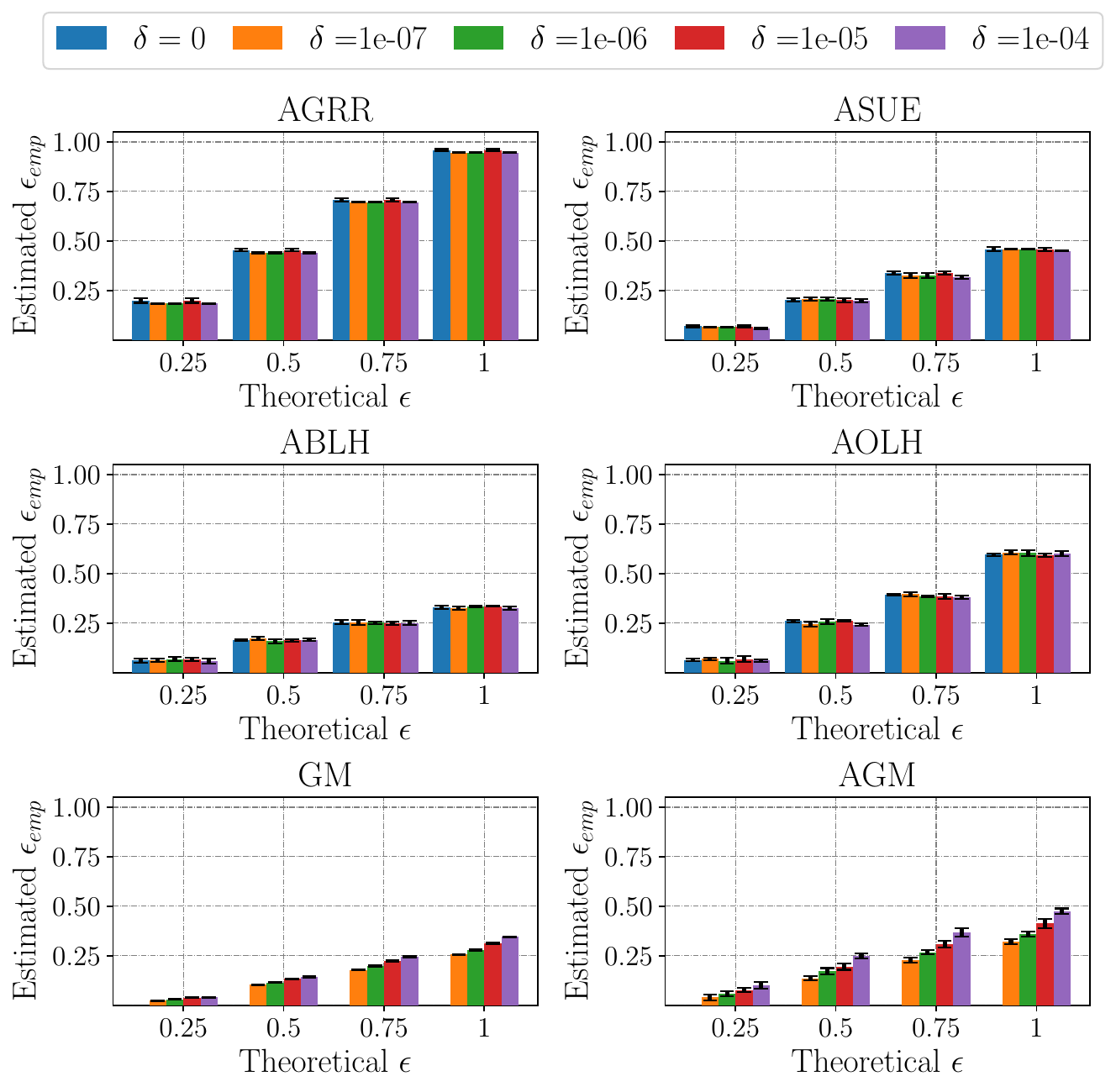}
        \caption{Domain size $k=100$.}
        \label{subfig:audit_delta_k100}
    \end{subfigure}
    \hfill
    \begin{subfigure}{\columnwidth}
        \includegraphics[width=1.0\linewidth]{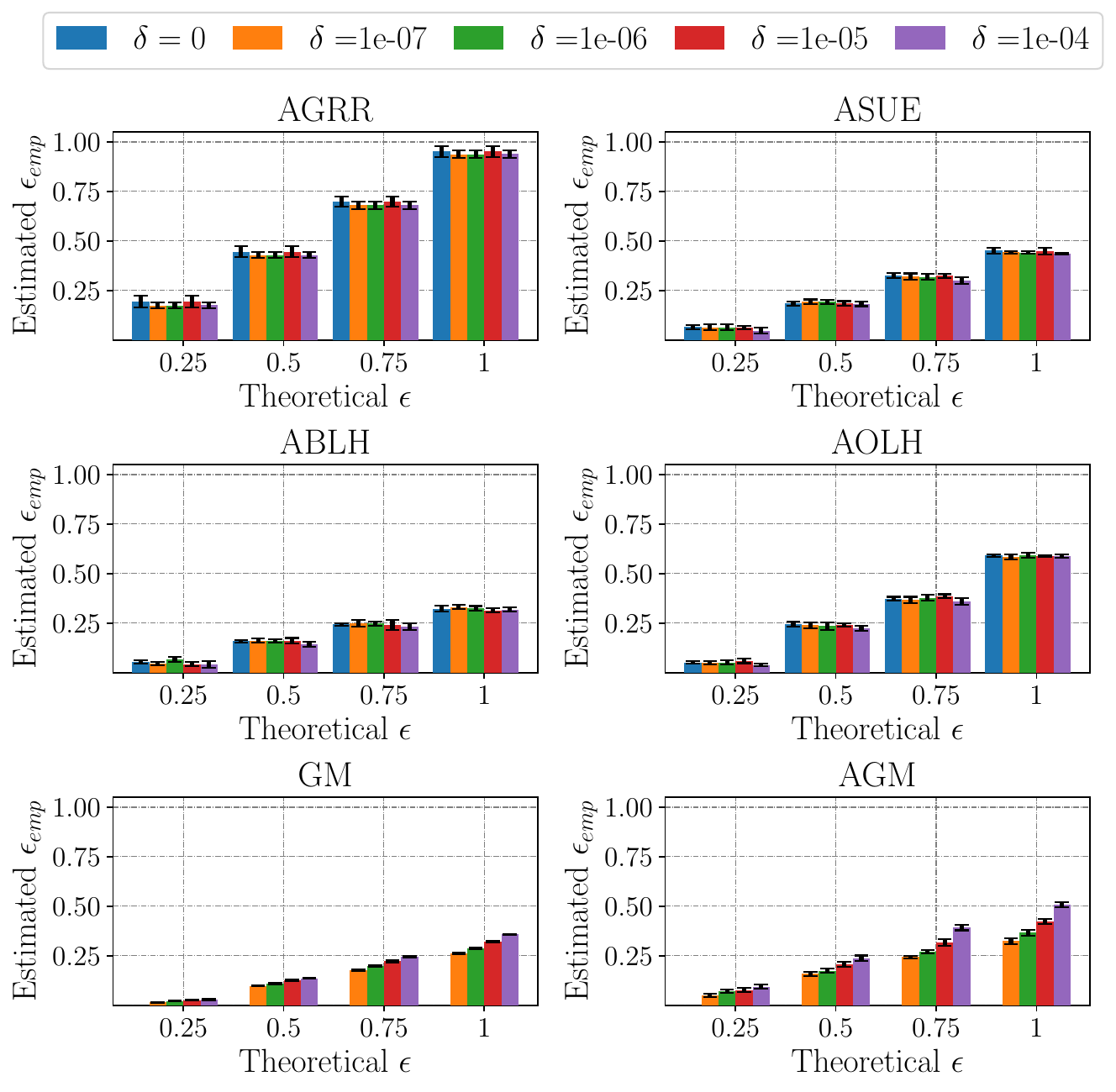}
        \caption{Domain size $k=150$.}
        \label{subfig:audit_delta_k150}
    \end{subfigure}
    \caption{Theoretical $\epsilon$ values (x-axis) versus estimated $\epslb$ values (y-axis) using our LDP-Auditor framework.
    We assess different privacy guarantees for six ($\epsdelta$)-LDP protocols across domain sizes $k \in \{100, 150\}$. 
    The special case $\delta=0$ corresponds to pure $\epsilon$-LDP, for which GM and AGM do not satisfy.}
    \label{fig:appendix_audit_delta}
\end{figure*}

\subsection{Case Study \#3: Auditing the LDP Sequential Composition in Longitudinal Studies} \label{app:add_exp_audit_seq_comp_long}

Following the experimental setup detailed in Section~\ref{sub:audit_seq_comp_long}, Figure~\ref{fig:add_exp_audit_seq_comp_long} illustrates the estimated $\epslb$ values (y-axis) for the eight $\epsilon$-LDP and both GM and AGM ($\epsdelta$)-LDP protocols according to the the number of data collections $\tau$ (x-axis), per report $\epsilon$ and domain size $k \in \{25, 50\}$, using our LDP-Auditor framework. 
Notice that a similar trend as in Figure~\ref{fig:audit_seq_comp_long} can be observed in Figure~\ref{fig:add_exp_audit_seq_comp_long}, for which the discussion in Section~\ref{sub:audit_seq_comp_long} is equally applicable to these results.

\begin{figure*}[!htb]
    \centering
    \begin{subfigure}{\columnwidth}
        \includegraphics[width=\linewidth]{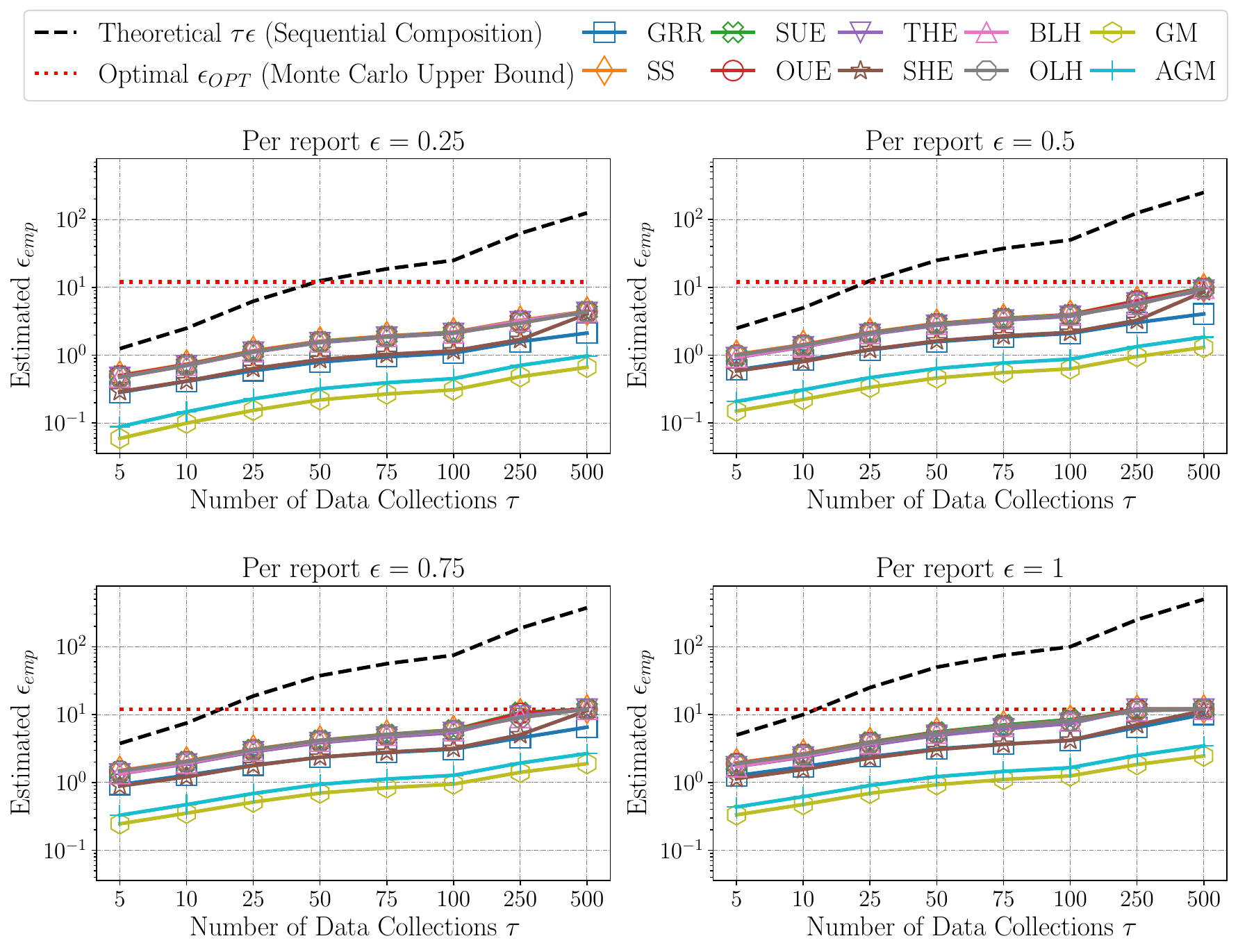}
        \caption{Domain size $k=25$.}
        \label{subfig:audit_seq_comp_long_k25}
    \end{subfigure}
    \hfill 
    \begin{subfigure}{\columnwidth}
        \includegraphics[width=\linewidth]{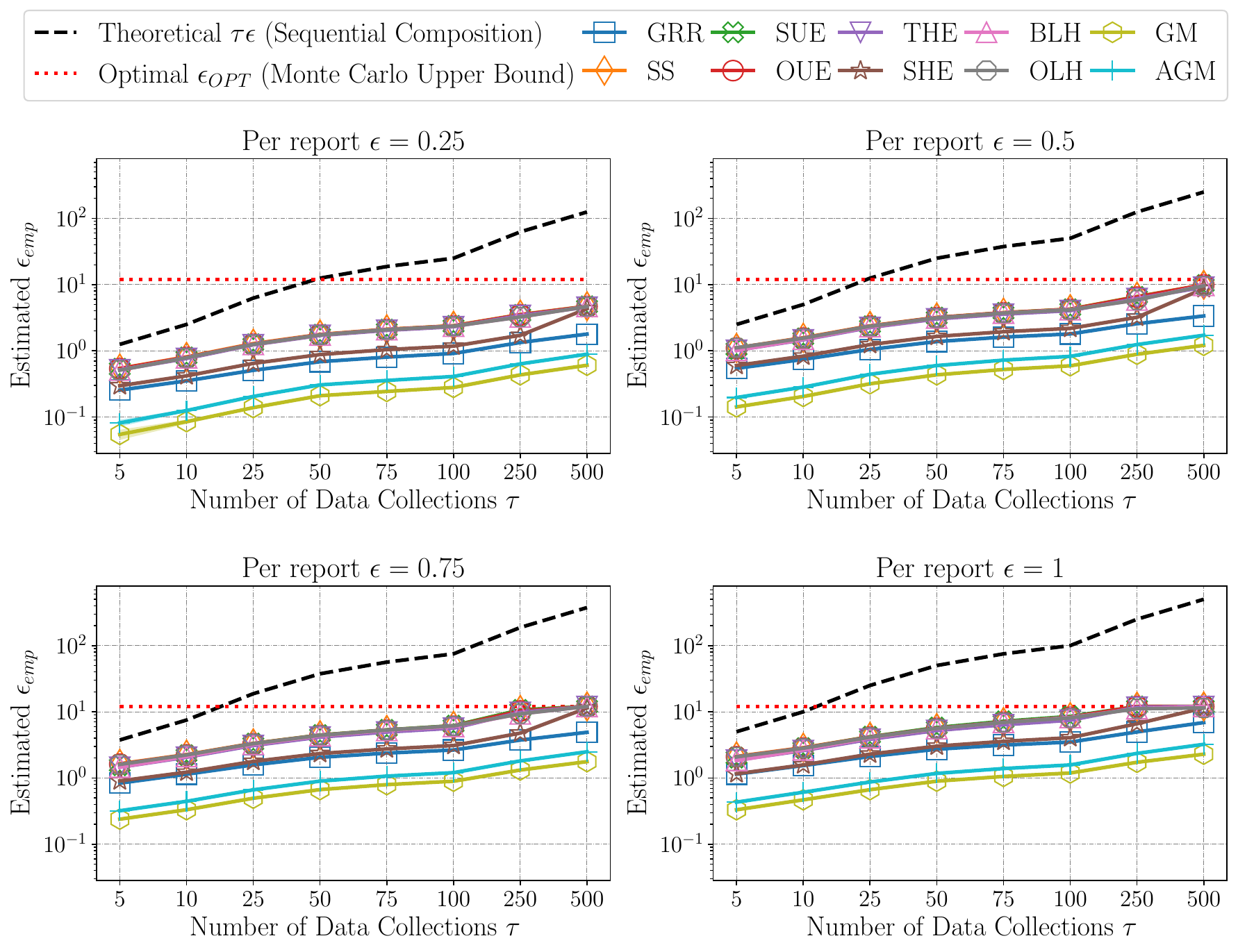}
        \caption{Domain size $k=50$.}
        \label{subfig:audit_seq_comp_long_k50}
    \end{subfigure}
    \caption{Estimated $\epslb$ (y-axis) versus the number of data collections $\tau$ (x-axis) using our LDP-Auditor framework for different domain sizes $k \in \{25, 50\}$. 
    We vary the per report $\epsilon$-LDP guarantee for the following LDP frequency estimation protocols: GRR, SS, SUE, OUE, BLH, OLH, SHE, THE, GM and AGM.
    For both approximate-LDP protocols, namely GM and AGM, $\delta=1e^{-5}$.}
    \label{fig:add_exp_audit_seq_comp_long}
\end{figure*}

\subsection{Case Study \#4: LDP Auditing with Multidimensional Data} \label{app:add_exp_audit_multidimensional}

Following the experimental setup detailed in Section~\ref{sub:audit_multidimensional}, Figure~\ref{fig:add_exp_audit_rs+fd} illustrates the comparison of theoretical $\epsilon$ values (x-axis) with estimated $\epslb$ values (y-axis) for the five RS+FD protocols, based on the number of attributes $d$ and domain size $k \in \{25, 50\}$, utilizing our LDP-Auditor framework.
Notice that a similar trend as in Figure~\ref{fig:audit_rs+fd} can be observed in Figure~\ref{fig:add_exp_audit_rs+fd}, for which the discussion in Section~\ref{sub:audit_multidimensional} is equally applicable to these results.

\begin{figure*}[!h]
    \centering
    \begin{subfigure}{1\columnwidth}
        \includegraphics[width=\linewidth]{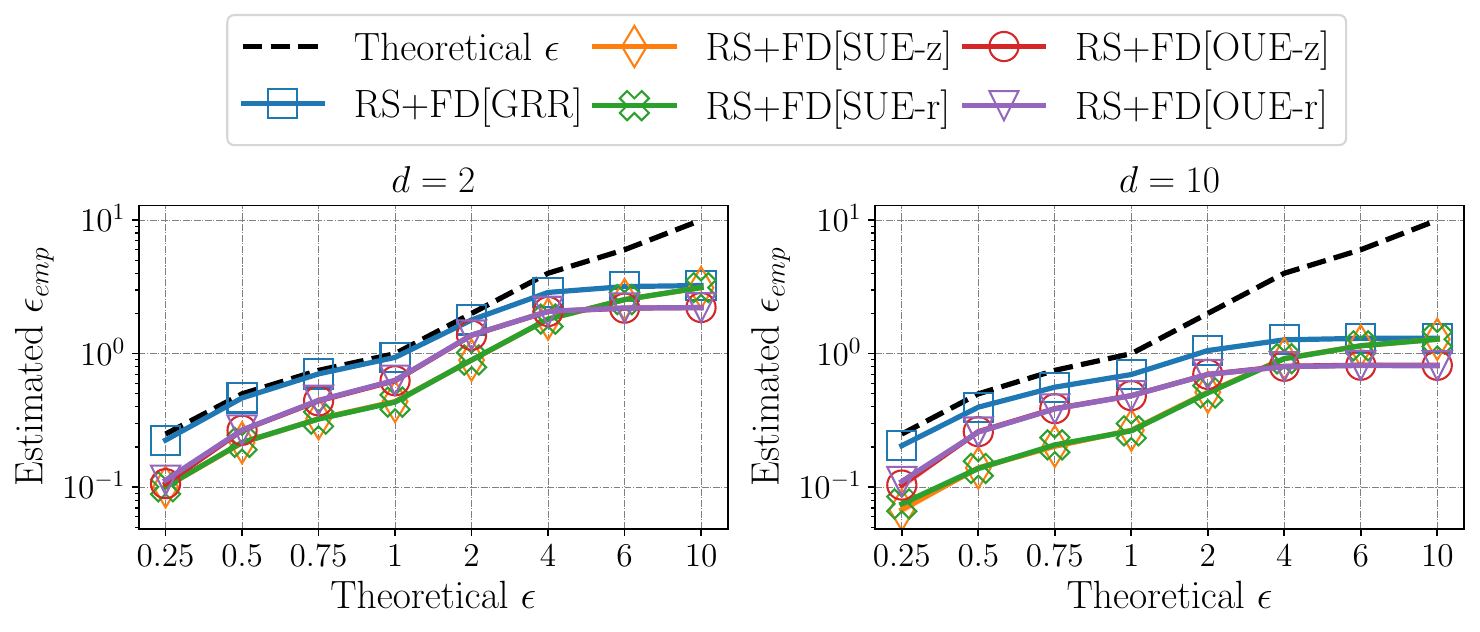}
        \caption{Domain size $k=25$.}
        \label{fig:audit_rs+fd_k25}
    \end{subfigure}
    \hfill 
    \begin{subfigure}{1\columnwidth}
        \includegraphics[width=\linewidth]{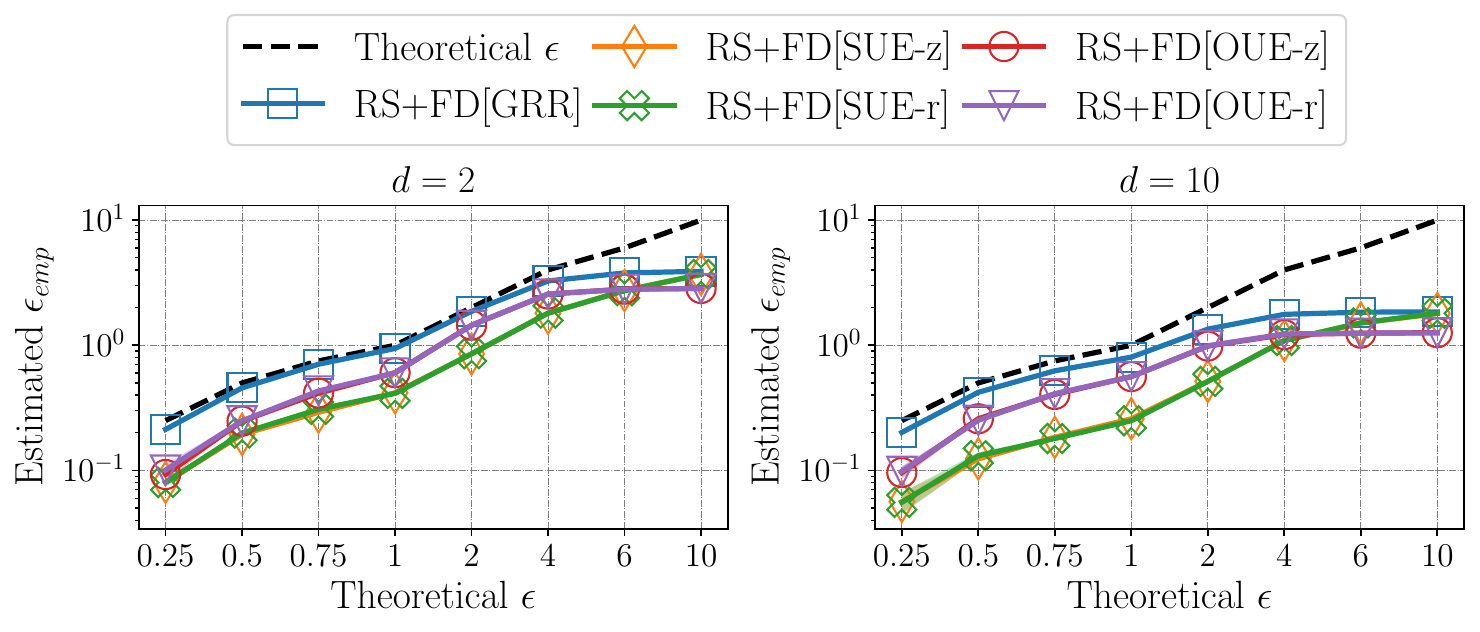}
        \caption{Domain size $k=50$.}
        \label{fig:audit_rs+fd_k50}
    \end{subfigure}
    \caption{Theoretical $\epsilon$ (x-axis) versus estimated $\epslb$ (y-axis) using our LDP-Auditor framework comparing different number of attributes $d$ for five RS+FD~\cite{Arcolezi_rs_fd} protocols with domain sizes $k=25$ and $k=50$.}
    \label{fig:add_exp_audit_rs+fd}
\end{figure*}

\end{document}